\documentclass{article}
\usepackage[utf8]{inputenc}

\usepackage{xcolor}

\usepackage{tikz}
\usepackage{subcaption}

\usepackage[english]{babel}

\usepackage[letterpaper,top=2cm,bottom=2cm,left=3cm,right=3cm,marginparwidth=1.75cm]{geometry}

\usepackage{amsmath,amsfonts,amssymb,amsthm}
\usepackage{mathrsfs}
\usepackage{graphicx}
\usepackage[colorlinks=true, allcolors=blue]{hyperref}

\definecolor{knownblue}{RGB}{35,87,145}
\definecolor{openorange}{RGB}{184,96,16}
\definecolor{excludedred}{RGB}{165,48,48}
\definecolor{panelgray}{RGB}{248,249,251}

\newcommand{\cat}{\ensuremath{\mathsf{Cat}\text{-}\mathrm{PCG}}}
\newcommand{\catp}[2]{\ensuremath{\mathsf{Cat}\text{-}\mathrm{PCG}(#1,#2)}}

\newtheorem{definition}{Definition}
\newtheorem{theorem}{Theorem}
\newtheorem{lemma}{Lemma}
\newtheorem{corollary}{Corollary}
\newtheorem{proposition}{Proposition}

\newtheorem{problem}{Problem}
\newtheorem{observation}{Observation}
\newcommand{\ed}[2]{\draw (v#1)--(v#2);}

\newcommand{\tnil}{\Lambda}

\newcommand{\onegraph}[1]{%
  \node[v] (v1) at (0,0) {\tiny $#1$};
}

\newcommand{\twograph}[2]{%
  \node[v] (v1) at (-0.35,0) {\tiny $#1$};
  \node[v] (v2) at (0.35,0)  {\tiny $#2$};
}

\newcommand{\threegraph}[3]{%
  \node[v] (v1) at (-0.45,-0.25) {\tiny $#1$};
  \node[v] (v2) at (0.45,-0.25)  {\tiny $#2$};
  \node[v] (v3) at (0,0.55)      {\tiny $#3$};
}

\newcommand{\fourgraph}[4]{%
  \node[v] (v1) at (-0.5,0.5)  {\tiny $#1$};
  \node[v] (v2) at (0.5,0.5)   {\tiny $#2$};
  \node[v] (v3) at (0.5,-0.5)  {\tiny $#3$};
  \node[v] (v4) at (-0.5,-0.5) {\tiny $#4$};
}

\newcommand{\fivegraph}[5]{%
  \node[v] (v1) at (90:0.65)  {\tiny $#1$};
  \node[v] (v2) at (162:0.65) {\tiny $#2$};
  \node[v] (v3) at (234:0.65) {\tiny $#3$};
  \node[v] (v4) at (306:0.65) {\tiny $#4$};
  \node[v] (v5) at (18:0.65)  {\tiny $#5$};
}

\newcommand{\gcell}[4]{%
\begin{minipage}[t]{0.18\textwidth}
\centering
\begin{tikzpicture}[
  scale=1,
  v/.style={circle,draw,inner sep=0.8pt,minimum size=4.2mm},
  every edge/.style={line width=0.35pt}
]
#2
#3
\end{tikzpicture}

{\scriptsize #1}\\[-1mm]
{\tiny \(I=#4\)}
\end{minipage}
}

\title{Pairwise Compatibility Graphs of Unweighted Caterpillars}
\author{Cyriac Antony, Manuel Lafond, Angelo Monti and Blerina Sinaimeri}

\begin{document}
\maketitle

\begin{abstract}
In this paper we study pairwise compatibility graphs whose witness tree is required to be an unweighted caterpillar.  This is a natural restriction of PCGs, in the same line of research as the study of star witnesses, which led to the class of double-threshold graphs. In the caterpillar case, we show that this restriction is equivalent to a one-dimensional difference-labelling problem: a graph admits such a representation if and only if its vertices can be assigned integer labels so that adjacency is determined exactly by whether the absolute difference of two labels belongs to a given interval. 

Using this characterization, we derive structural properties of the class, including bounded twin-width, and study its relation with familiar graph families. We show that the class contains cycles and two-dimensional grids, while it excludes natural obstructions such as the subdivided star $S_{2,2,2,2,2}$, the cube graph $Q_3$, the wheel $W_6$, and the complement of the $3$-sunlet. We also consider restricted representations, including centipede witnesses and representations with a given interval and we give a polynomial-time recognition algorithm when the interval is fixed.

\end{abstract}

\section{Introduction}
Pairwise compatibility graphs (PCGs) were introduced in~\cite{KPM03} in the context of phylogenetic reconstruction and have since been studied as a graph class defined through distances in a tree; see, for example, the surveys~\cite{Calamoneri_survey_2025, Rahman2020}. Intuitively, in a PCG representation the vertices of the graph are placed at the leaves of an edge-weighted tree, called a witness tree, and two vertices are adjacent when the distance between the corresponding leaves belongs to a prescribed interval.

Several variants of PCGs have been considered by imposing restrictions on either the distance interval or the witness tree. Restrictions on the interval lead, for instance, to leaf powers, min-leaf powers, and exact leaf powers~\cite{B10,B10exact,L22}. Restrictions on the topology of the witness tree have also led to important graph classes. A well-known example is given by star witnesses introduced in~\cite{CALAMONERI2013}: in this case, distances between leaves are determined by sums of pendant edge weights, and the corresponding graphs coincide with double-threshold graphs~\cite{Xiao2020,Kobayashi22,Jamison20}. Caterpillar witnesses form another natural restriction and have appeared in the study of several graph families admitting PCG representations, including cycles, ladder graphs, interval graphs, grids, and circular graphs~\cite{calamoneri2014pairwise,Y10,Calamoneri_survey_2025}.

In this paper we study PCGs whose witness tree is required to be an unweighted caterpillar. We call this graph class $\mathrm{Cat\text{-}PCG}$. This combines two restrictions: the topology of the witness tree is fixed to be a caterpillar, and all its edges have unit weight. The latter condition is meaningful only when the topology is prescribed. Indeed, for unrestricted PCGs, integer weights can be eliminated by subdividing edges, but this operation does not preserve the original topology of the witness tree~\cite{calamoneri_etal2013}. Thus the class of PCGs admitting an unweighted caterpillar witness is not simply a reformulation of the general PCG model.

The reason why this class can be studied effectively is that distances in an unweighted caterpillar have a very simple form. Once the leaves are attached to the spine, the distance between two leaves depends only on the difference between their attachment positions. This leads naturally to interval difference labellings: vertices are placed on the integer line, and adjacency is determined by whether the absolute difference of two labels belongs to a fixed interval. We show that a graph admits an unweighted caterpillar witness if and only if it admits such a difference labelling. This characterization is the main tool used throughout the paper. In this sense, unweighted caterpillar witnesses provide a difference-based analogue of the sum-based representations arising from star witnesses.

Using this viewpoint, we investigate the structure of the resulting class. We prove that every graph admitting an unweighted caterpillar witness has bounded twin-width. This is in contrast with treewidth and clique-width, which are unbounded for the class, since two-dimensional grids admit such representations. We also study the relation with familiar graph families. In particular, we show that cycles and two-dimensional grids belong to the class. Note that the latter result is a strengthening of~\cite[Theorem 4]{hakim2022newresults}, where grid graphs are shown to be PCGs of \emph{weighted} caterpillars. Thus the unweighted restriction still allows graph families of unbounded treewidth and clique-width.

On the other hand, several natural graphs are not in $\mathrm{Cat\text{-}PCG}$, including the subdivided star $S_{2,2,2,2,2}$, the cube graph $Q_3$, the wheel $W_6$, and the complement of the $3$-sunlet. These results show that the class has a non-trivial boundary: it is large enough to contain graph families of unbounded treewidth and clique-width, but still constrained by the one-dimensional nature of the representation. In particular, even some small and structurally simple graphs fail to admit an unweighted caterpillar witness.

We also consider more restricted forms of the model. We characterize the graphs admitting centipede witnesses, where each spine vertex has exactly one leaf, and we study the case of single-point intervals. These results show that further restrictions on the caterpillar topology or on the distance interval lead to explicit and structured subclasses.

Finally, we address the recognition problem for a given interval. While the recognition of $\mathrm{Cat\text{-}PCG}$ in general remains open, we show that, for fixed integers $a,b$, membership in $\mathrm{Cat\text{-}PCG}(a,b)$ can be decided in polynomial time. Polynomial-time recognition is known for star-PCGs, which coincide with double-threshold graphs and can be recognized in linear time~\cite{Kobayashi22}; our result provides a polynomial-time recognition algorithm for unweighted caterpillar witnesses when the distance interval is fixed.

The paper is organized as follows. Section~\ref{sec:preliminaries} contains the definitions and notation used throughout the paper. In Section~\ref{sec:difference_labellin} we introduce interval difference labellings and relate them to unweighted caterpillar witnesses; we also prove the bounded twin-width result. Section~\ref{sec:graph_classes} is devoted to examples and obstructions. In Section~\ref{sec:restrictions} we study restricted representations, including centipede witnesses and restrictions on the admissible interval. In Section~\ref{sec:recognition} we consider the recognition problem for a prescribed interval. Section~\ref{sec:graph_operations} studies closure properties of $\mathrm{Cat\text{-}PCG}$, and Section~\ref{sec:concl} concludes with open problems.

\section{Preliminaries}\label{sec:preliminaries}
For a graph $G$, we denote by $V(G)$ and $E(G)$ its vertex set and edge set, respectively. 
For a vertex $v\in V(G)$, we define $N_G(v)=\{u: \{u,v\}\in E(G)\}$  the \emph{neighborhood} of $v$ and we omit the subscript when the graph is clear from the context. Two distinct vertices \(u,v\in V(G)\) are called \emph{false twins} if $N(u)=N(v)$, and \emph{true twins} if $N(u) \setminus \{v\}=N(v) \setminus \{u\}$ and $uv \in E(G)$. We say that \(u\) and \(v\) are \emph{twins} if they are either true twins or false twins. For a tree $T$, we denote by $L(T)$ the set of leaves of $T$, that is the set of vertices of degree 1. 

\begin{definition}[\cite{KPM03}]\label{def.PCG}
A \emph{pairwise compatibility graph} (PCG) is a graph $G$ for which there exist a weighted tree $T$, whose edges have positive real weights, and an interval $I$ of nonnegative real numbers, such that
$V(G)=L(T)$ and, for every two distinct leaves $u,v\in L(T)$,
$uv\in E(G)$  if and only if $d_T(u,v)\in I$ (where $d_T(u,v)$ denotes the weighted distance between $u$ and $v$ in $T$).
In this case, $T$ is called a witness tree of $G$ and we denote $G=PCG(T,I).$
\end{definition}

\noindent
When all edges of $T$ have unit weight, we say that $T$ is \emph{unweighted}, and $d_T(u,v)$ denotes the number of edges in the unique path between $u$ and $v$ in $T$. A \emph{caterpillar} is a tree $ T $ such that removing leaves of $ T $ results in a path. We call this path the \emph{spine} of $ T $. 


\begin{definition}[\cat]\label{def:cat_PCGs}
For non-negative integers $0\leq a\leq b$, we denote by $\catp{a}{b}$ the class of graphs $G$ for which $G=\mathrm{PCG}(T,[a,b])$ for some unweighted caterpillar $T$.
We define
\[
\mathrm{Cat\mbox{-}PCG}
   =\bigcup_{0\leq a\leq b}\catp{a}{b}.
\]
\end{definition}

Let $T$ be an unweighted caterpillar witness for $G$. For every
$v\in V(G)$, we denote by $s_T(v)$ the unique neighbour of $v$ in $T$;
this vertex lies on the spine of $T$. If the spine vertices are ordered from left to right as $s_1,s_2,\ldots,s_k$, we define $\pi_T(v)$ to be the position on the spine where the leaf $v$ is attached. That is, $\pi_T(v)=i$ whenever $s_T(v)=s_i$. In particular, for distinct $u,v\in V(G)$, $d_T(u,v)=|\pi_T(u)-\pi_T(v)|+2.$

We denote by $[n]=\{1,2,\ldots, n\}$. Throughout the paper, $\mathbb{N}=\{1,2,\ldots\}$.

The $n \times m$ \emph{grid} is the graph whose vertex set is $\{ (i, j) : i \in [n], j \in [m] \}$,  and two vertices $(i,j)$ and $(k,\ell)$ are adjacent if and only if
$|i-k|+|j-\ell|=1$. The \emph{cube graph} \(Q_3\) is the graph whose vertex set is \(\{0,1\}^3\), where two vertices are adjacent if and only if they differ in exactly one coordinate.

For every $n\geq 3$ we denote by $C_n$ the cycle on $n$ vertices. For every $n\geq 4$, we denote by $W_n$ the \emph{wheel graph}  on $n$ vertices (i.e., the graph obtained by adding a universal vertex to the cycle $C_{n-1}$).

We denote by $G\dot\cup H$ the disjoint union of $G$ and $H$, and by
$\overline{G}$ the complement of $G$.  For two graphs \(H_1,H_2\) on the same vertex set, with \(E(H_2)\subseteq E(H_1)\), we denote by \(H_1-H_2\) the graph with edge set \(E(H_1)\setminus E(H_2)\).

\section{Difference labellings and structural properties}\label{sec:difference_labellin}
Motivated by two notions of difference graphs~\cite{hammer_etal1990,haray1990sum}, for non-negative integers $ a $ and $ b $ with $ a\leq b $, we define an \emph{$ [a,b] $-difference labelling} of a graph $ G $ as follows.

\begin{definition}[${[a,b]}$-difference labelling]\label{def:ab-labelling}
Let $a, b$ be two non-negative integers with $ 0\leq a\leq b $. An \emph{$ [a,b] $-difference labelling} of $ G $ is a function $ p\colon V(G)\to \mathbb{N} $ such that $ uv\in E(G) $ if and only if $ |p(u)-p(v)|\in [a,b] $. 
\end{definition}

To the best of our knowledge, $[a,b] $-difference labellings have not been studied before.  Notice that in a caterpillar witness tree,  the distance between two leaves is determined by the absolute difference of their spine positions and this is the basic reason why Cat-PCGs can be studied through  $[a,b]$-difference labelling as shown by the following theorem.  

\begin{theorem}\label{thm:char}
Let $G$ be a graph and let $a,b$ be two non-negative integers with $0\leq a\leq b$. Then $G\in \catp{a+2}{b+2}$  if and only if $G$ admits an $[a,b]$-difference labelling.
\end{theorem}
\begin{proof}
Suppose first $G=PCG(T,[a+2,b+2])$ then for any two distinct vertices $u,v$ we have that $uv\in E(G)$ if and only if $|\pi_T(u)-\pi_T(v)|\in [a,b]$ which means that $\pi_T$ is an $[a,b]$-difference labelling. 

Conversely, given an $[a,b]$-difference labelling $p$, let $c=\min_{v\in V(G)} p(v)$ and define for every $v$, $q(v)=p(v)-c+1$. This is just a shifting of the labels so they start from $1$. Notice that $q$ is again an $[a,b]$-labelling as $|q(v)-q(u)|=|p(v)- p(u)|$.  Now, let $r=\max_{v\in V(G)} q(v)$,and take a path$ s_1,\dots,s_r$, and attach each vertex $v\in V(G)$ as a leaf to $s_{q(v)}$. The resulting tree is a caterpillar and satisfies $d_T(u,v)=|p(u)-p(v)|+2$ for all distinct $u,v$. Hence it is a witness for $G$ with interval $[a+2,b+2]$.
\end{proof}

We have the following.

\begin{observation}\label{obs:same_label_twins}
Let \(G\) have an \([a,b]\)-difference labelling \(p\). If \(p(u)=p(v)\) for two distinct vertices \(u,v\), then \(u\) and \(v\) are twins. More precisely, they are false twins if \(a>0\), and true twins if \(a=0\).  In particular, if $G$ is twin-free then $p$ is injective.
\end{observation}

\begin{proposition}\label{prop:close}
Let \(0\leq a\leq b\). The class of graphs admitting an
\([a,b]\)-difference labelling is closed under:
\begin{enumerate}
    \item taking induced subgraphs;
    \item taking disjoint unions;
    \item adding false twins, when \(a>0\);
    \item adding true twins, when \(a=0\).
\end{enumerate}
\end{proposition}

\begin{proof}
An induced subgraph inherits the restriction of the original
labelling. For the disjoint union of two graphs, translate all labels of one component sufficiently far to the right so that every
difference between vertices of distinct components is larger than
\(b\). Finally, a twin of a vertex \(u\) can be assigned the same label as \(u\). The two vertices then have the same neighbours outside the pair. Their mutual difference is zero, so they are nonadjacent when \(a>0\) and adjacent when \(a=0\).
\end{proof}

Consequently, if two graphs belong to
\(\mathrm{Cat\text{-}PCG}(a+2,b+2)\) for the same pair \(a,b\),
then their disjoint union also belongs to
\(\mathrm{Cat\text{-}PCG}(a+2,b+2)\).


We finish this section with a structural consequence of the difference-labelling characterization.  We will see in Section~\ref{sec:graph_classes}  that two-dimensional grids are Cat-PCGs.
Since square grids have unbounded treewidth and unbounded clique-width
\cite{GOLUMBIC2000}, it follows that Cat-PCG has unbounded
treewidth and unbounded clique-width. Nevertheless, in the next theorem we show that \(\cat\) has bounded twin-width.

We use the following equivalent formulation of twin-width~\cite{Bonnet2022}. 
A contraction sequence of a graph \(G\) is a sequence of partitions of \(V(G)\), starting from the partition into singletons and ending with the single part \(V(G)\), where each step is obtained by merging two parts. 

For a partition \(\mathcal P\), two distinct parts \(X,Y\in\mathcal P\) are called \emph{mixed} if the bipartite graph between \(X\) and \(Y\) is neither empty nor complete. The \emph{width} of \(\mathcal P\) is the maximum, over all parts \(X\in\mathcal P\), of the number of parts mixed with \(X\). The \emph{width} of a contraction sequence is the maximum width of any partition appearing in the sequence. The \emph{twin-width} of \(G\) is the minimum width of a contraction sequence of \(G\).

\begin{theorem}\label{thm:tww_cat}
Every graph in \(\cat\) has twin-width at most \(10\).
\end{theorem}

\begin{proof}
Let \(G'\in\cat\). By Theorem~\ref{thm:char}, \(G'\) admits an
\([a,b]\)-difference labelling $p$.

If two vertices receive the same label, then they are twins by
Observation~\ref{obs:same_label_twins}. Merging vertices with the same label
does not create any mixed pair, since all vertices in such a class have the same
adjacency to every vertex outside the class. Therefore, it is enough to consider
the case where $p$ is injective.

By shifting all labels if necessary, assume that \(p(v)\geq 1\) for every
\(v\in V(G')\). Let \(n\) be a power of two such that
\[
n\geq \max\{p(v):v\in V(G')\}.
\]
Let \(G\) be the graph with vertex set \([n]\), where
\[
ij\in E(G)
\quad\Longleftrightarrow\quad
|i-j|\in [a,b].
\]
Then \(G'\) is an induced subgraph of \(G\). Since twin-width does not increase
when taking induced subgraphs, it is enough to give a contraction sequence for
\(G\) of width at most \(10\).

We use a contraction sequence that repeatedly merges consecutive pairs of
intervals. At every step, the current parts are intervals of \([n]\). Initially,
\[
\mathcal I=([1],[2],\ldots,[n]).
\]
In the first round, we merge consecutive pairs \([1]\) with \([2]\), \([3]\)
with \([4]\), and so on. In the next round, we merge consecutive pairs of the
new intervals, and we continue in this way until only one interval remains.

Thus, at the beginning of round \(t\), every interval has size \(2^{t-1}\).
During round \(t\), we process the intervals from left to right. Hence, at any
intermediate step of the round, the current interval partition consists of
intervals of size \(2^t\) in the processed prefix and intervals of size
\(2^{t-1}\) in the unprocessed suffix.

We show that every current interval has mixed degree at most \(10\). Consider
an intermediate partition
\[
\mathcal I^*=(I_1,\ldots,I_\ell),
\]
and fix an interval
\[
I_c=[q,r].
\]
Let \(\delta=r-q\). Hence
\[
|I_c|=\delta+1\in\{2^{t-1},2^t\}.
\]

First consider intervals to the right of \(I_c\). Let \(I_d\) be such an
interval. If all distances between vertices of \(I_c\) and vertices of \(I_d\)
are smaller than \(a\), then there is no edge between \(I_c\) and \(I_d\). If
all such distances lie in \([a,b]\), then all possible edges between \(I_c\) and
\(I_d\) are present. If all such distances are larger than \(b\), then again
there is no edge between \(I_c\) and \(I_d\).

Therefore \(I_c\) and \(I_d\) can be mixed only when \(I_d\) intersects one of
the two boundary zones
\[
[q+a,r+a]
\qquad\text{or}\qquad
[q+b,r+b].
\]
Each of these zones has size \(|I_c|\leq 2^t\), while every current interval has
size at least \(2^{t-1}\). Hence each boundary zone intersects at most three
intervals of \(\mathcal I^*\). Thus \(I_c\) has at most \(6\) mixed neighbours
to its right.

Now consider intervals to the left of \(I_c\). The same argument shows that
mixed neighbours to the left can occur only among intervals intersecting one of
the two boundary zones
\[
[q-a,r-a]
\qquad\text{or}\qquad
[q-b,r-b].
\]
We claim that each of these zones intersects at most two intervals to the left.
If \(|I_c|=2^{t-1}\), then the zone has size \(2^{t-1}\), and every current
interval has size at least \(2^{t-1}\). If \(|I_c|=2^t\), then \(I_c\) lies in
the processed prefix, and all intervals to its left also have size \(2^t\),
because contractions are performed from left to right. Thus, in either case,
each boundary zone contributes at most two intervals. Therefore \(I_c\) has at
most \(4\) mixed neighbours to its left.

Finally, every current interval has mixed degree at most $6+4=10$. Thus this contraction sequence has width at most \(10\). Hence every graph in
\(\cat\) has twin-width at most \(10\).
\end{proof}

\section{Connections with known graph classes}\label{sec:graph_classes}

\subsection{Graphs that are Cat-PCGs}\label{subsec:catPCGs}
Here we show that some simple class graphs are subsets of \cat. A graph $G$ is a \emph{unit interval graph} if there exists a family
$\{I_v : v\in V(G)\}$ of intervals of length one on the real line such that, for every two distinct vertices $u,v$,
$uv\in E(G)$ if and only if $I_u\cap I_v\neq\emptyset.$  

Notice first that  in \cite{brandsadt2008ptolemaic} it is proved that (unweighted) caterpillar leaf power graphs (i.e., the union of \catp{2}{b} for all integers $b$) are exactly unit interval graphs. So we have the following.

\begin{theorem}\cite{brandsadt2008ptolemaic}\label{theo:unit_inerval}
Every unit interval graph is in \cat. 
\end{theorem}

\noindent
From Theorem~\ref{theo:unit_inerval} we have that paths are clearly \cat. 
We can extend this to caterpillars (note that in the next subsection, we show that not every tree is in \cat).

\begin{theorem}
Every caterpillar is in \cat.
\end{theorem}

\begin{proof}
Let $T$ be a caterpillar.
The cases \(|V(T)|\leq 2\) are immediate, so assume $|V(T)| \geq 3$. Assume also that the
spine of \(T\) is the path $s_1,s_2,\ldots,s_k$ ordered from left to right. For every \(i\in[k]\), let \(L_i\) be the
set of leaves adjacent to \(s_i\). Set $M=2k+1$ and, for every \(i\in[k]\), define $c_i=k-i+1$. Thus, $c_1=k,\quad c_2=k-1,\quad\ldots,\quad c_k=1$. Define
$$
p(v)=
\begin{cases}
1+(i-1)M, \textrm{ if $v=s_i$ where $s_i$ is a spine vertex, }\\
1+iM-c_i, \textrm{ if $v=x$ where  \(x\in L_i\).}
\end{cases}
$$
\noindent
We prove that \(p\) is a \([k+1,M]\)-difference labelling of \(T\). First, consider the edges of \(T\). If \(s_i s_{i+1}\) is a spine
edge, then  $|p(s_{i+1})-p(s_i)|=M \in [k+1,M]$. Next, for any edge $(x, s_i)$, with $x \in L_i$, we have  $|p(x)-p(s_i)|=M-c_i \in [k+1,M]$ as $1 \leq c_i \leq k $.

We now show that no non-edge has label difference in this interval. If \(s_i\) and \(s_j\) are not consecutive spine vertices, then
$|p(s_i)-p(s_j)|=|i-j|M\geq 2M>M$. If \(x,y\in L_i\), then \(p(x)=p(y)\), and hence $|p(x)-p(y)|=0<k+1$. Next, let \(x\in L_i\), and consider a spine vertex \(s_j\neq s_i\).
If \(j=i+1\), then $|p(s_{i+1})-p(x)|=c_i\leq k$. If \(j\leq i-1\), then  $|p(x)-p(s_j)|=(i-j+1)M-c_i \geq 2M-k>M$. If \(j\geq i+2\), then $|p(s_j)-p(x)|=(j-i-1)M+c_i>M$. Finally, let \(x\in L_i\) and \(y\in L_j\), where \(i<j\). Then $p(y)-p(x)=(j-i)M+c_i-c_j$.
Since $c_i-c_j=j-i$, we obtain $p(y)-p(x)=(j-i)(M+1)>M$.

Thus \(p\) is a \([k+1,2k+1]\)-difference labelling of \(T\).
By Theorem~1, $T\in \mathrm{Cat\text{-}PCG}(k+3,2k+3)$.
\end{proof}

\noindent
Since $K_{1,3}$ is a caterpillar but is not a unit interval graph,
the class of caterpillar graphs is not contained in the class of
unit interval graphs.


\begin{theorem}\label{theo:cycles}
For every $n\ge 3$, the cycle $C_n$ is a \cat. More specifically,
$C_3\in\catp{2}{2} $, and for every $ n\ge4$,
 $C_n\in\catp{5}{7}$.
\end{theorem}

\begin{proof}
The statement for $C_3$ is immediate, since $C_3$ is complete and thus
admits a $[0,0]$-difference labelling.

Let $n\ge4 $. By Theorem~\ref{thm:char}, it suffices to construct a
$[3,5]$-difference labelling of $C_n$.

Assume first that $n=2k$ is even. Let the vertices of $C_n$ be
$(v_0,v_1,\ldots,v_{2k-1})$ in cyclic order and define
$$
p(v_i)=
\begin{cases}
1, & i=0,\\
4i, & 1\le i\le k-1,\\
4k+1, & i=k,\\
4(2k-i)+2, & k+1\le i\le 2k-1.
\end{cases}
$$

Thus the labels around the cycle are:
$$
1,4,8,\ldots,4k-4,4k+1,4k-2,4k-6,\ldots,6.
$$

Consecutive vertices on the cycle have label difference $3$, $4$, or $5$,
and hence every edge of $C_n$ corresponds to a difference in $[3,5]$.

To prove the converse, observe that all labels except $1$ and $4k+1$
belong to one of the arithmetic progressions
$$
4,8,\ldots,4k-4,
\qquad
6,10,\ldots,4k-2,
$$
both having common difference $4$.
Therefore, for two distinct labels different from $1$ and $4k+1$, a
difference in $[3,5]$ can only be equal to $4$, which occurs precisely
between consecutive terms of one of the two progressions. By construction,
these pairs are exactly edges of the cycle.

Moreover, the only labels at distance in $[3,5]$ from $1$ are $4$ and
$6$, while the only labels at distance in $[3,5]$ from $4k+1$ are
$4k-4$ and $4k-2$. These are exactly the four neighbours of the
corresponding vertices on the cycle.

Hence
$$
v_iv_j\in E(C_n)
\quad\Longleftrightarrow\quad
|p(v_i)-p(v_j)|\in[3,5],
$$
and $p$ is a $[3,5]$-difference labelling of $C_n$.

The odd case $n=2k+1$ is identical. Define
$$
p(v_i)=
\begin{cases}
1, & i=0,\\
4i, & 1\le i\le k,\\
4k+3, & i=k+1,\\
4(2k+1-i)+2, & k+2\le i\le 2k.
\end{cases}
$$

The labels are
$$
1,4,8,\ldots,4k,4k+3,4k-2,4k-6,\ldots,6.
$$

Again, consecutive vertices have differences $3$, $4$, or $5$. As before,
all labels except $1$ and $4k+3$ belong to two arithmetic progressions of
difference $4$, so the only pairs with difference in $[3,5]$ are
consecutive terms of one of these progressions together with the pairs
$$
{1,4},\ {1,6},\ {4k+3,4k},\ {4k+3,4k-2}.
$$
These are precisely the edges of the cycle.

Therefore $p$ is a $[3,5]$-difference labelling of $C_n$, and hence
$C_n\in\catp{5}{7}$.
\end{proof}

Next, we show that (2-dimensional) grids are Cat-PCGs.  It is convenient to first consider a skewed version. The $n\times m$
\emph{skewed grid} is the graph with vertex set
$$
\bigcup_{i\in[n]}\{(i,n-i+1),(i,n-i+2),\ldots,(i,n-i+m)\},
$$
with the same adjacency rule:
$(i,j)$ and $(k,\ell)$ are adjacent if and only if
$|i-k|+|j-\ell|=1$.

Figure~\ref{fig:skewed-grid} gives a $[5,6]$-difference labelling of the
$4\times 5$ skewed grid. The next lemma generalizes this construction and
shows that every $n\times m$ skewed grid admits an
$[n+1,n+2]$-difference labelling.

\begin{figure}[hbt] 
\centering
\includegraphics[width=.5\textwidth]{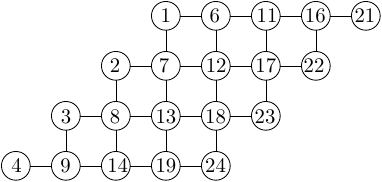}
\caption{%
    A $ [5,6] $-difference labelling of the $ 4\times 5 $ skewed grid. 
}
\label{fig:skewed-grid}
\end{figure}

\begin{lemma}\label{lem:skewed grid}
For any $n, m\in \mathbb{N}$, the $n \times m$ skewed grid is a \catp{n+3}{n+4}.
\end{lemma}
\begin{proof}
Let $G$ be the $n\times m$ skewed grid. For a vertex $(i,j)$, set
$r=j-n+i$. Notice that $r\in[m]$. Define $p(i,j)=(r-1)(n+1)+i.$
Equivalently, $p(i,j)=(j-n+i-1)(n+1)+i$. We claim that $p$ is an $[n+1,n+2]$-difference labelling. Let
$(i,j)$ and $(i',j')$ be two distinct vertices, and set
$r=j-n+i$, $r'=j'-n+i'$.
Assume without loss of generality that $p(i',j')>p(i,j)$. Then
$$
p(i',j')-p(i,j)=(r'-r)(n+1)+(i'-i).
$$
Since $i,i'\in[n]$, we have
$$
-(n-1)\leq i'-i\leq n-1.
$$
\noindent
If $p(i',j')-p(i,j)=n+1$, then necessarily $r'-r=1$ and $i'-i=0$. Hence $i'=i$ and $r'=r+1$, which means $j'=j+1$. Thus the two vertices are adjacent in the same row.

\noindent
Otherwise if $ p(i',j')-p(i,j)=n+2$, then necessarily $r'-r=1$ and $i'-i=1$. Hence $i'=i+1$, and since
$j'=n-i'+r'=n-(i+1)+(r+1)=n-i+r=j$,
the two vertices are adjacent in the same column.

Conversely, if two vertices are adjacent in the same row, say $(i,j)$ and $(i,j+1)$, then their labels differ by $n+1$. If two vertices are adjacent in the same column, say $(i,j)$ and $(i+1,j)$, then their labels differ by $n+2$.

Therefore two vertices are adjacent exactly when their label difference belongs to
$[n+1,n+2]$. Hence $p$ is an $[n+1,n+2]$-difference labelling of $G$.
\end{proof}

Since $ n\times m $ grid is an induced subgraph of $ n\times (n+m) $ skewed grid, the next theorem holds.
\begin{theorem}
For every $n, m \in \mathbb{N}$, the $n \times m$ grid is a \catp{n+3}{n+4}.
\end{theorem}


\begin{figure}[p]
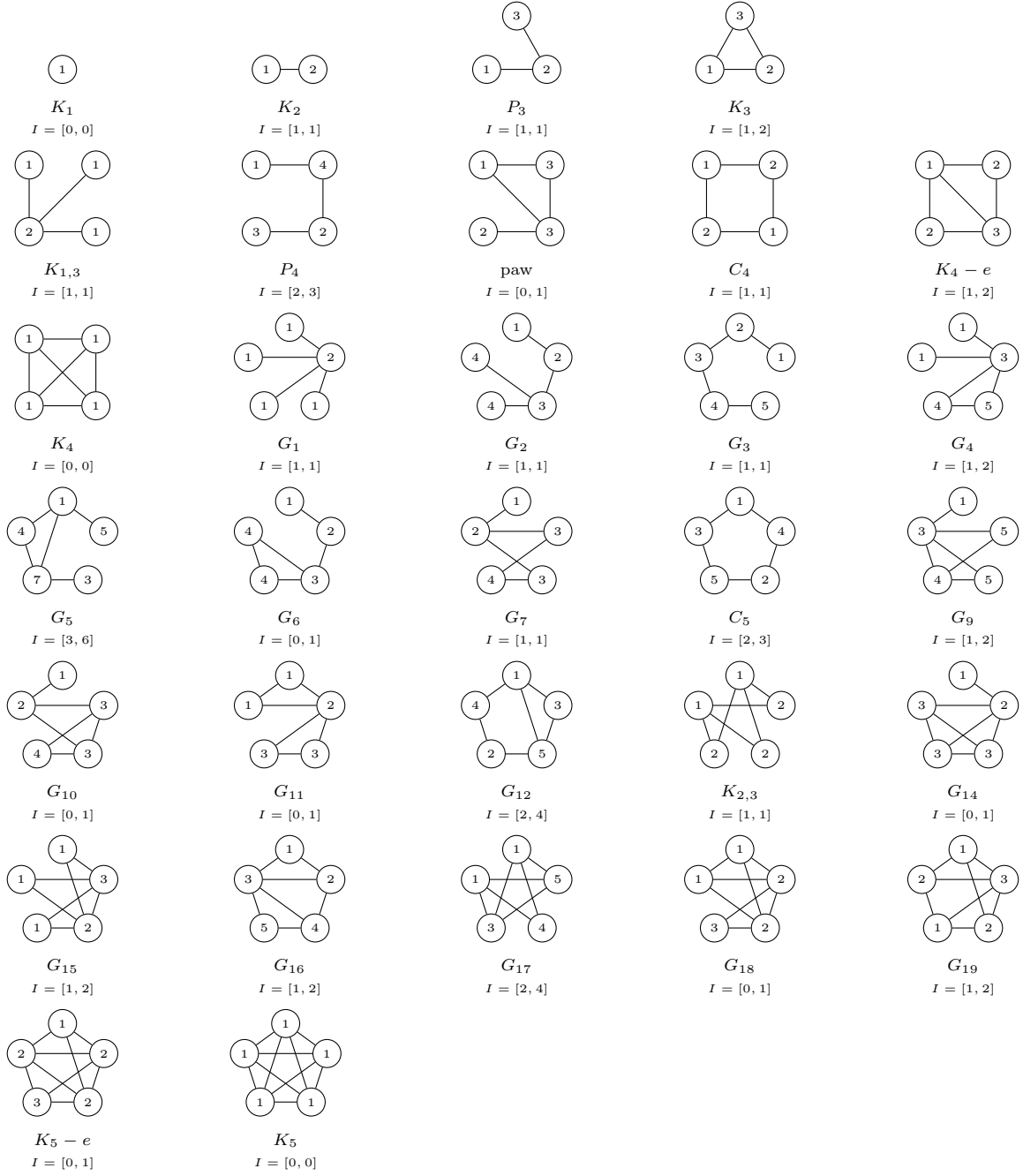

\centering

\begin{tabular*}{\textwidth}{@{\extracolsep{\fill}}ccccc@{}}
\gcell{\(K_1\)}{\onegraph{1}}{}{[0,0]}
&
\gcell{\(K_2\)}{\twograph{1}{2}}{\ed{1}{2}}{[1,1]}
&
\gcell{\(P_3\)}{\threegraph{1}{2}{3}}{\ed{1}{2}\ed{2}{3}}{[1,1]}
&
\gcell{\(K_3\)}{\threegraph{1}{2}{3}}{\ed{1}{2}\ed{1}{3}\ed{2}{3}}{[1,2]}
&
\medskip
\\[2ex]

\gcell{\(K_{1,3}\)}{\fourgraph{1}{1}{1}{2}}{\ed{1}{4}\ed{2}{4}\ed{3}{4}}{[1,1]}
&
\gcell{\(P_4\)}{\fourgraph{1}{4}{2}{3}}{\ed{1}{2}\ed{2}{3}\ed{3}{4}}{[2,3]}
&
\gcell{paw}{\fourgraph{1}{3}{3}{2}}{\ed{1}{2}\ed{1}{3}\ed{2}{3}\ed{3}{4}}{[0,1]}
&
\gcell{\(C_4\)}{\fourgraph{1}{2}{1}{2}}{\ed{1}{2}\ed{2}{3}\ed{3}{4}\ed{4}{1}}{[1,1]}
&
\gcell{\(K_4-e\)}{\fourgraph{1}{2}{3}{2}}{\ed{1}{2}\ed{1}{3}\ed{1}{4}\ed{2}{3}\ed{3}{4}}{[1,2]}
\medskip
\\[2ex]

\gcell{\(K_4\)}{\fourgraph{1}{1}{1}{1}}{\ed{1}{2}\ed{1}{3}\ed{1}{4}\ed{2}{3}\ed{2}{4}\ed{3}{4}}{[0,0]}
&
\gcell{\(G_1\)}{\fivegraph{1}{1}{1}{1}{2}}{\ed{1}{5}\ed{2}{5}\ed{3}{5}\ed{4}{5}}{[1,1]}
&
\gcell{\(G_2\)}{\fivegraph{1}{4}{4}{3}{2}}{\ed{1}{5}\ed{2}{4}\ed{3}{4}\ed{4}{5}}{[1,1]}
&
\gcell{\(G_3\)}{\fivegraph{2}{3}{4}{5}{1}}{\ed{1}{2}\ed{1}{5}\ed{2}{3}\ed{3}{4}}{[1,1]}
&
\gcell{\(G_4\)}{\fivegraph{1}{1}{4}{5}{3}}{\ed{1}{5}\ed{2}{5}\ed{3}{4}\ed{3}{5}\ed{4}{5}}{[1,2]}
\medskip
\\[2ex]

\gcell{\(G_5\)}{\fivegraph{1}{4}{7}{3}{5}}{\ed{1}{2}\ed{1}{3}\ed{1}{5}\ed{2}{3}\ed{3}{4}}{[3,6]}
&
\gcell{\(G_6\)}{\fivegraph{1}{4}{4}{3}{2}}{\ed{1}{5}\ed{2}{3}\ed{2}{4}\ed{3}{4}\ed{4}{5}}{[0,1]}
&
\gcell{\(G_7\)}{\fivegraph{1}{2}{4}{3}{3}}{\ed{1}{2}\ed{2}{4}\ed{2}{5}\ed{3}{4}\ed{3}{5}}{[1,1]}
&
\gcell{\(C_5\)}{\fivegraph{1}{3}{5}{2}{4}}{\ed{1}{2}\ed{2}{3}\ed{3}{4}\ed{4}{5}\ed{5}{1}}{[2,3]}
&
\gcell{\(G_9\)}{\fivegraph{1}{3}{4}{5}{5}}{\ed{1}{2}\ed{2}{3}\ed{2}{4}\ed{2}{5}\ed{3}{4}\ed{3}{5}}{[1,2]}
\medskip
\\[2ex]

\gcell{\(G_{10}\)}{\fivegraph{1}{2}{4}{3}{3}}{\ed{1}{2}\ed{2}{4}\ed{2}{5}\ed{3}{4}\ed{3}{5}\ed{4}{5}}{[0,1]}
&
\gcell{\(G_{11}\)}{\fivegraph{1}{1}{3}{3}{2}}{\ed{1}{2}\ed{1}{5}\ed{2}{5}\ed{3}{4}\ed{3}{5}\ed{4}{5}}{[0,1]}
&
\gcell{\(G_{12}\)}{\fivegraph{1}{4}{2}{5}{3}}{\ed{1}{2}\ed{1}{4}\ed{1}{5}\ed{2}{3}\ed{3}{4}\ed{4}{5}}{[2,4]}
&
\gcell{\(K_{2,3}\)}{\fivegraph{1}{1}{2}{2}{2}}{\ed{1}{3}\ed{1}{4}\ed{1}{5}\ed{2}{3}\ed{2}{4}\ed{2}{5}}{[1,1]}
&
\gcell{\(G_{14}\)}{\fivegraph{1}{3}{3}{3}{2}}{\ed{1}{5}\ed{2}{3}\ed{2}{4}\ed{2}{5}\ed{3}{4}\ed{3}{5}\ed{4}{5}}{[0,1]}
\medskip
\\[2ex]
\gcell{\(G_{15}\)}{\fivegraph{1}{1}{1}{2}{3}}{\ed{1}{4}\ed{1}{5}\ed{2}{4}\ed{2}{5}\ed{3}{4}\ed{3}{5}\ed{4}{5}}{[1,2]}
&
\gcell{\(G_{16}\)}{\fivegraph{1}{3}{5}{4}{2}}{\ed{1}{2}\ed{1}{5}\ed{2}{3}\ed{2}{4}\ed{2}{5}\ed{3}{4}\ed{4}{5}}{[1,2]}
&
\gcell{\(G_{17}\)}{\fivegraph{1}{1}{3}{4}{5}}{\ed{1}{3}\ed{1}{4}\ed{1}{5}\ed{2}{3}\ed{2}{4}\ed{2}{5}\ed{3}{5}}{[2,4]}
&
\gcell{\(G_{18}\)}{\fivegraph{1}{1}{3}{2}{2}}{\ed{1}{2}\ed{1}{4}\ed{1}{5}\ed{2}{4}\ed{2}{5}\ed{3}{4}\ed{3}{5}\ed{4}{5}}{[0,1]}
&
\gcell{\(G_{19}\)}{\fivegraph{1}{2}{1}{2}{3}}{\ed{1}{2}\ed{1}{4}\ed{1}{5}\ed{2}{3}\ed{2}{5}\ed{3}{4}\ed{3}{5}\ed{4}{5}}{[1,2]}
\medskip
\\[2ex]
\gcell{\(K_5-e\)}{\fivegraph{1}{2}{3}{2}{2}}{\ed{1}{2}\ed{1}{4}\ed{1}{5}\ed{2}{3}\ed{2}{4}\ed{2}{5}\ed{3}{4}\ed{3}{5}\ed{4}{5}}{[0,1]}
&
\gcell{\(K_5\)}{\fivegraph{1}{1}{1}{1}{1}}{\ed{1}{2}\ed{1}{3}\ed{1}{4}\ed{1}{5}\ed{2}{3}\ed{2}{4}\ed{2}{5}\ed{3}{4}\ed{3}{5}\ed{4}{5}}{[0,0]}
\medskip
\\[2ex]
\end{tabular*}

\caption{All non-isomorphic connected graphs on at most five vertices. The numbers inside the vertices give an \([a,b]\)-difference labelling, and the corresponding interval \(I=[a,b]\) is shown below each graph.}
\label{fig:connected_graphs_at_most_five}
\end{figure}

We now look at the class \(\cat\) on small graphs. We show that all the graphs with at most $5$ vertices are \cat. We will use the following simple closure property to reduce the disconnected small graphs to the connected ones.

\begin{corollary}\label{cor:add_small_components}
Let $0\leq a\leq b$ with $b\geq 2$, and let
$G\in\catp{a}{b}$. Then
$G\dot\cup K_1$, $G\dot\cup K_2$ and $G\dot\cup \overline{K_2}$
also belong to $\catp{a}{b}$.
\end{corollary}

\begin{proof}
Since the distance between any two distinct leaves of an unweighted tree is at least $2$, if $a<2$ then
$\catp{a}{b}=\catp{2}{b}$. Thus, assume $a\geq 2$ and set $A=a-2$ and $B=b-2$.
By Theorem~\ref{thm:char}, $G$ admits an $[A,B]$-difference labelling. Clearly $K_1$ admits an $[A,B]$-difference labelling. If $A>0$, label the vertices of $K_2$ by $1,1+A$ and those of $\overline{K_2}$ by $1,1$. If $A=0$, use respectively the labels $1,1$ and $1,B+2$.

Thus $K_1$, $K_2$, and $\overline{K_2}$ admit an $[A,B]$-difference labelling. The claim follows from Proposition~\ref{prop:close} and Theorem~\ref{thm:char}.
\end{proof}

\begin{theorem}\label{theo:five_vertices}
Every graph on at most five vertices is in \cat.
\end{theorem}

\begin{proof}
The connected graphs on at most five vertices are shown in Figure~\ref{fig:connected_graphs_at_most_five},
together with explicit difference labellings. Let $G$ be disconnected. If $G$ has a connected component $H$ with at least three vertices, then $G-H$ has at most two vertices. Hence $G$ is obtained from $H$ by adding $K_1$, $K_2$, or
$\overline{K_2}$, and the result follows from
Corollary~\ref{cor:add_small_components}. Otherwise, every component of $G$ is either $K_1$ or $K_2$. Since both admit a $[1,1]$-difference labelling, repeated application of Proposition~\ref{prop:close} shows that $G$ also admits a $[1,1]$-difference labelling. Hence $G\in\cat$ by
Theorem~\ref{thm:char}.
\end{proof}

\subsection{Graphs that are not Cat-PCGs}\label{subsec:not_catPCGS}
In this section we prove that the graphs shown in Figure~\ref{fig:graphs-section4} are not Cat-PCGs.  These examples show, in particular, that not every tree is a \cat, that the positive result for two-dimensional grids does not extend to higher-dimensional grids, and that the smallest graphs that are not in \cat~have six vertices.

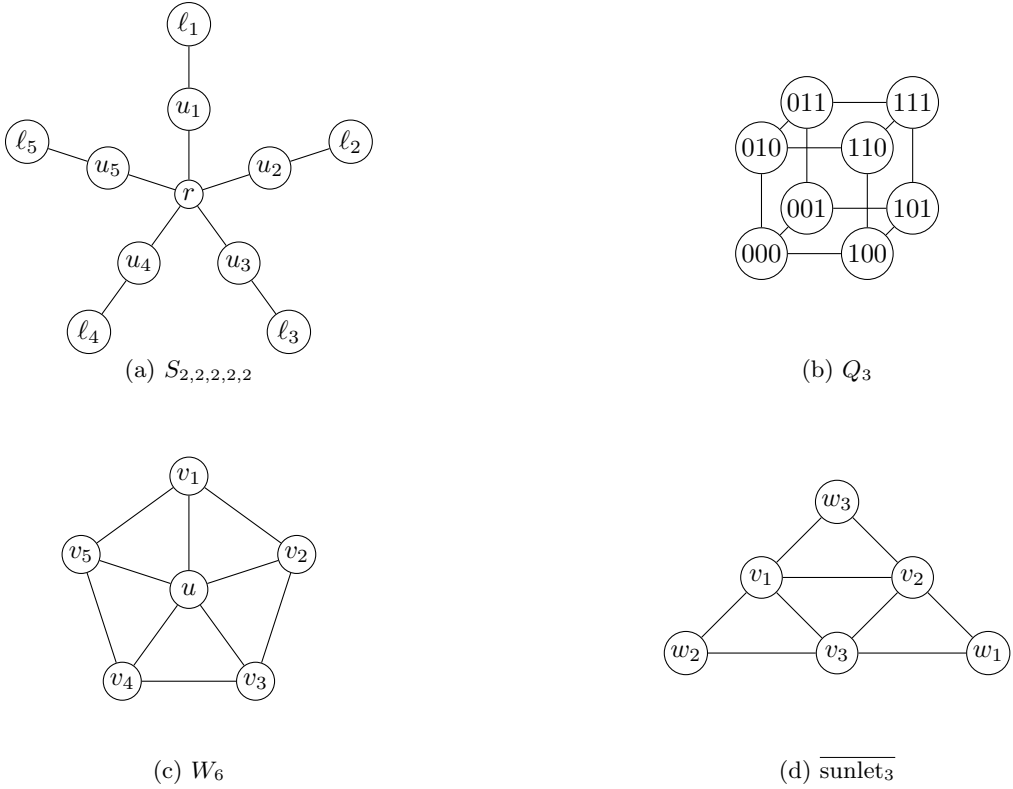
\begin{figure}[!ht]
\centering

\begin{subfigure}[b]{0.45\textwidth}
\centering
\begin{minipage}[c][4.3cm][c]{\linewidth}
\centering
\begin{tikzpicture}[scale=0.75, every node/.style={circle, draw, inner sep=1.5pt}]
  \node (r) at (0,0) {$r$};

  \node (u1) at (90:1.5) {$u_1$};
  \node (u2) at (18:1.5) {$u_2$};
  \node (u3) at (-54:1.5) {$u_3$};
  \node (u4) at (-126:1.5) {$u_4$};
  \node (u5) at (162:1.5) {$u_5$};

  \node (l1) at (90:3) {$\ell_1$};
  \node (l2) at (18:3) {$\ell_2$};
  \node (l3) at (-54:3) {$\ell_3$};
  \node (l4) at (-126:3) {$\ell_4$};
  \node (l5) at (162:3) {$\ell_5$};

  \draw (r) -- (u1) -- (l1);
  \draw (r) -- (u2) -- (l2);
  \draw (r) -- (u3) -- (l3);
  \draw (r) -- (u4) -- (l4);
  \draw (r) -- (u5) -- (l5);
\end{tikzpicture}
\end{minipage}
\caption{$S_{2,2,2,2,2}$}
\end{subfigure}
\hfill
\begin{subfigure}[b]{0.45\textwidth}
\centering
\begin{minipage}[c][4.3cm][c]{\linewidth}
\centering
\begin{tikzpicture}[scale=1,
    every node/.style={circle, draw, inner sep=1.2pt, minimum size=5mm}]
    
    \node (a1) at (0,0) {$000$};
    \node (a2) at (1.4,0) {$100$};
    \node (a3) at (1.4,1.4) {$110$};
    \node (a4) at (0,1.4) {$010$};

    \node (b1) at (0.6,0.6) {$001$};
    \node (b2) at (2.0,0.6) {$101$};
    \node (b3) at (2.0,2.0) {$111$};
    \node (b4) at (0.6,2.0) {$011$};

    \draw (a1)--(a2)--(a3)--(a4)--(a1);
    \draw (b1)--(b2)--(b3)--(b4)--(b1);
    \draw (a1)--(b1);
    \draw (a2)--(b2);
    \draw (a3)--(b3);
    \draw (a4)--(b4);
\end{tikzpicture}
\end{minipage}
\caption{$Q_3$}
\end{subfigure}

\vspace{1em}

\begin{subfigure}[b]{0.45\textwidth}
\centering
\begin{minipage}[c][4.3cm][c]{\linewidth}
\centering
\begin{tikzpicture}[scale=1,
    every node/.style={circle, draw, inner sep=1.2pt, minimum size=5mm}]
    
    \node (u) at (0,0) {$u$};
    \node (v1) at (90:1.5) {$v_1$};
    \node (v2) at (18:1.5) {$v_2$};
    \node (v3) at (-54:1.5) {$v_3$};
    \node (v4) at (-126:1.5) {$v_4$};
    \node (v5) at (162:1.5) {$v_5$};

    \draw (v1)--(v2)--(v3)--(v4)--(v5)--(v1);
    \draw (u)--(v1);
    \draw (u)--(v2);
    \draw (u)--(v3);
    \draw (u)--(v4);
    \draw (u)--(v5);
\end{tikzpicture}
\end{minipage}
\caption{$W_6$}
\end{subfigure}
\hfill
\begin{subfigure}[b]{0.45\textwidth}
\centering
\begin{minipage}[c][4.3cm][c]{\linewidth}
\centering
\begin{tikzpicture}[
    scale=1,
    every node/.style={circle, draw, inner sep=1.2pt, minimum size=5.5mm}
]
\node (w3) at (0,2) {$w_3$};
\node (v1) at (-1,1) {$v_1$};
\node (v2) at (1,1) {$v_2$};
\node (w2) at (-2,0) {$w_2$};
\node (v3) at (0,0) {$v_3$};
\node (w1) at (2,0) {$w_1$};

\draw (w3) -- (v1);
\draw (w3) -- (v2);
\draw (v1) -- (v2);
\draw (v1) -- (w2);
\draw (v1) -- (v3);
\draw (v2) -- (v3);
\draw (v2) -- (w1);
\draw (w2) -- (v3);
\draw (v3) -- (w1);
\end{tikzpicture}
\end{minipage}
\caption{$\overline{\mathrm{sunlet}_3}$}
\end{subfigure}
\caption{Graphs discussed in Section~4. The subdivided star $S_{2,2,2,2,2}$ is shown in (a). The graphs in (b), (c), and (d) are the cube graph $Q_3$, the wheel $W_6$, and the complement of $\mathrm{sunlet}_3$.}
\label{fig:graphs-section4}
\end{figure}

In the following proofs we will use two types of obstructions. 

\begin{lemma}[Betweenness obstruction]\label{lem:betweenness_obstruction}
Let $G$ have an $[a,b]$-difference labelling $p$, and let
$x,u,v,w\in V(G)$ be distinct vertices such that $xu,xw\in E(G)$. If either
\[
p(x)\leq p(u)<p(v)<p(w),
\]
or
\[
p(u)<p(v)<p(w)\leq p(x),
\]
then \(xv\in E(G)\).
\end{lemma}


\begin{proof}
We prove the first case; the second follows by reversing the order of the line. Since \(xu,xw\in E(G)\), we have
$a\leq p(u)-p(x)$  and $p(w)-p(x)\leq b.$
Moreover, $p(u)-p(x)\leq p(v)-p(x)\leq p(w)-p(x).$
Thus
$a\leq p(v)-p(x)\leq b,$ and hence \(xv\in E(G)\).
\end{proof}

From this lemma we have the following straightforward corollary that will be used to simplify later proofs.
\begin{corollary}\label{cor:corollary_betweeness}
Let \(p(u)<p(v)<p(w)\). If a vertex \(x\) is adjacent to \(u\)
and \(w\), but not to \(v\), then
\[
p(u)<p(x)<p(w).
\]
\end{corollary}

\begin{lemma}[Neighbourhood inclusion]\label{lem:inclusion_neigborhood}
Let $G=(V,E)$ have an $ [a,b] $-difference labelling $p$. Let $u_1,u_2,u_3\in V$ be distinct vertices such that $p(u_1)\le p(u_2)\le p(u_3)$ and  $p(u_3)-p(u_1)\le b-a$.  Then $N(u_2)\subseteq N(u_1)\cup N(u_3).$
\end{lemma}

\begin{proof}
For simplicity we first set:
\[
\alpha=p(u_2)-p(u_1),\qquad \beta=p(u_3)-p(u_2).
\]
Then $\alpha,\beta\ge 0$, and by assumption $\alpha+\beta=p(u_3)-p(u_1)\le b-a.$

Now, let $x\in N(u_2)$. Then $a\le |p(x)-p(u_2)|\le b.$ We need to consider the following cases:

\paragraph{Case I: $p(x)\ge p(u_2)$.}  Then  let $d=p(x)-p(u_2)$ and clearly  $a\le d \le b$. Now 

\begin{itemize}
    \item Suppose $d\ge a+\beta$. Then $a\le d-\beta \leq d$  but as 
$d-\beta=p(x)-p(u_2)+p(u_2)-p(u_3)=p(x)-p(u_3)$, we have that $p(x)-p(u_3) \in [a,b]$ thus, $x\in N(u_3)$.
    \item Suppose $d<a+\beta$.  Then $a \le d+\alpha < a+\alpha+\beta\le a+(b-a)=b$ but $d+\alpha= p(x)-p(u_2)+p(u_2)-p(u_1)=p(x)-p(u_1)$. Thus, we have that $p(x)-p(u_1) \in [a,b]$ thus, $x\in N(u_1)$.
\end{itemize}

\paragraph{Case II: $p(x) < p(u_2)$.}  Then  let $d=p(u_2)-p(x)$ and again we have  $a\le d \le b$. Now,

\begin{itemize}
    \item Suppose $d\ge a+\alpha$. Then $a\le d-\alpha\le b$ but as $d-\alpha= p(u_2)-p(x)+p(u_1)-p(u_2)=p(u_1)-p(x)$ we have $p(u_1)-p(x) \in [a,b]$ thus, $x\in N(u_1)$.
    \item Suppose $d<a+\alpha$.  Then $a\le d+\beta<a+\alpha+\beta\le a+(b-a)\le b$ but as $d+\beta=p(u_2)-p(x) + p(u_3)-p(u_2)= p(u_3)-p(x)$. Thus, we have that $p(u_3) -p(x) \in [a,b]$ thus, $x\in N(u_3)$.
\end{itemize}

Therefore every neighbor $x$ of $u_2$ belongs to
$N(u_1)\cup N(u_3)$. Hence $N(u_2)\subseteq N(u_1)\cup N(u_3).$
\end{proof}

\begin{theorem}\label{theo:tree_cat_pcgs}
The subdivided star $S_{2,2,2,2,2}$, depicted in Figure~\ref{fig:graphs-section4}(a), is not a Cat-PCG.
\end{theorem}
\begin{proof}
Suppose, to the contrary, that $S_{2,2,2,2,2}$ admits an $[a,b]$-difference labelling $p$. Consider the center $r$, whose neighbors are $u_1,\dots,u_5$. At least three of these neighbors lie on the same side of $p(r)$. Assume
$ p(r)\leq p(u_1)\leq p(u_2)\leq p(u_3)$ (the other case  $p(u_1)\leq p(u_2)\leq p(u_3) \leq p(r)$ is symmetric).  Since $r$ is adjacent to $u_1,u_2,u_3$, we have $a\leq p(u_i)-p(r)\leq b$ for all $i=1,2,3$.  In particular, $p(u_3) - p(r) \leq b$ and $p(u_1) - p(r) \geq a$. Hence $p(u_3)-p(u_1)\leq b-a$. Thus, the vertices $u_1,u_2$ and $u_3$ satisfy the conditions of Lemma~\ref{lem:inclusion_neigborhood} and thus it must hold that $N(u_2)\subseteq N(u_1)\cup N(u_3).$
But $u_2$ has a private leaf $\ell_2$, adjacent to $u_2$ and to neither $u_1$ nor $u_3$ reaching a contradiction.
\end{proof}

\begin{corollary}\label{cor:trees}
Not every tree is a Cat-PCG.
\end{corollary}

In Section~\ref{subsec:catPCGs} we showed that all two-dimensional grids are in \cat.  The next result shows that this does not extend to three-dimensional grids, since the cube graph is not a Cat-PCG. 

\begin{theorem}\label{thm:cube_not_cat}
The cube graph \(Q_3\) depicted in Figure~\ref{fig:graphs-section4}(b) is not in \(\cat\).
\end{theorem}

\begin{proof}
Suppose, on the contrary that \(Q_3\) admits an \([a,b]\)-difference labelling \(p\). 
Since \(Q_3\) is twin-free, Observation~\ref{obs:same_label_twins} implies that \(p\) is injective. Let \(u\) be the vertex with minimum label. Let
$N(u)=\{v_1,v_2,v_3\}$, where $p(v_1)<p(v_2)<p(v_3).$

We first claim that no vertex different from \(v_2\) can have a label strictly between \(p(v_1)\) and \(p(v_3)\). Indeed, suppose that there exists a vertex \(y\neq v_2\) such that
$p(v_1)<p(y)<p(v_3).$ Because the inequalities are strict and $y \neq v_1, v_3$ and by assumption $y\neq v_2$ we have \(y\notin \{v_1,v_2,v_3\}=N(u)\). Since $u$ has the minimum label then $p(u)< p(v_1) < p(y) < p(v_3)$ and $u$ is adjacent to $v_1$ and $v_3$. Thus, we can apply Lemma~\ref{lem:betweenness_obstruction} with \(u\)
as the distinguished vertex and \(v_1,y,v_3\) as the ordered
triple and we obtain $uy \in E(Q_3)$ which contradicts $y \notin N(u)$.

We now use the structure of the cube. Represent the vertices of
\(Q_3\) by binary strings of length three, with two vertices adjacent if and only if they differ in exactly one coordinate. We may assume that $u=000$. Therefore, $\{v_1,v_2,v_3\}=\{100,010,001\}$, where the indices \(1,2,3\) are determined by the ordering $p(v_1)<p(v_2)<p(v_3)$. Let \(x\) be the vertex obtained by taking the bitwise OR of the binary strings corresponding to \(v_1\) and \(v_3\). Then \(x\) is adjacent to both \(v_1\) and \(v_3\), but not to \(v_2\). Since $p(v_1)<p(v_2)<p(v_3)$, Corollary~\ref{cor:corollary_betweeness} implies that
$p(v_1)<p(x)<p(v_3)$. Moreover, \(x\neq v_2\), contradicting the claim above.

Therefore, $Q_3\notin\mathrm{Cat\text{-}PCG}$.
\end{proof}

As any \(d\)-dimensional grid with at least three side lengths greater than \(1\) contains an induced copy of the cube graph \(Q_3\), the following holds.

\begin{corollary}\label{cor:high_dimensional_grids}
Let \(G\) be a \(d\)-dimensional grid. If at least three side lengths of \(G\) are greater than \(1\), then \(G\notin\cat\).
\end{corollary}


Calamoneri et al.~\cite{calamoneri2014pairwise} showed that
$W_n$, for $7\leq n\leq11$, does not admit a caterpillar witness and  thus clearly is not a \cat. Moreover, wheels $W_n$ with $n\geq9$ are not PCGs
\cite{Baiocchi2019,BeaudouinLafon2019}. We now show directly that $W_6$ is not in Cat-PCG, thereby settling the only remaining case.

\begin{theorem}\label{theo:w6}
The wheel graph \(W_6\) depicted in Figure~\ref{fig:graphs-section4}(c) is not in \(\cat\).
\end{theorem}

\begin{proof}
Let \(u\) be the center of \(W_6\), and let the other five vertices induce a
cycle \(C_5\). Suppose, on the contrary that \(W_6\) admits an
\([a,b]\)-difference labelling \(p\). Since \(W_6\) is twin-free, from 
Observation~\ref{obs:same_label_twins} we have that \(p\) is injective.

Since \(u\) is adjacent to the five vertices of the cycle, at least three of
them lie on the same side of \(p(u)\). Wlog let
\(v_1,v_2,v_3\) be the first three cycle vertices to the right of \(u\) (the case when they are on the left is symmetric), so that
$$
p(u)<p(v_1)<p(v_2)<p(v_3).
$$
Let \(r,s\) be the two remaining cycle vertices. By the choice of
\(v_1,v_2,v_3\), neither \(r\) nor \(s\) lies in the interval
\((p(u),p(v_3))\).

Notice that subgraph induced by \(\{v_1,v_2,v_3\}\) in the cycle \(C_5\) has either one edge or two edges.

\medskip
\noindent
\textbf{Case 1.} Suppose \(C_5[\{v_1,v_2,v_3\}]\) has exactly one edge.
Then one among \(v_1,v_2,v_3\) is isolated in this induced subgraph. In a \(5\)-cycle, each of the two remaining vertices \(r,s\) is adjacent to this isolated vertex and to exactly one endpoint of the unique edge. Hence there exist
\(x\in\{r,s\}\) and indices \(i<j\) such that
$$
xv_i\notin E(W_6)
\qquad\text{and}\qquad
xv_j\in E(W_6).
$$
Clearly,  \(x\) is also adjacent to \(u\). Applying
Corollary~\ref{cor:corollary_betweeness} to the triple
\((u,v_i,v_j)\) and the vertex \(x\), we obtain
$$
p(u)<p(x)<p(v_j)\leq p(v_3),
$$
contradicting the  fact that neither \(r\) nor \(s\) lies in the interval
\((p(u),p(v_3))\).

\medskip
\noindent
\textbf{Case 2.} Suppose \(C_5[\{v_1,v_2,v_3\}]\) has exactly two edges.
This means that these three vertices induce a path. In a \(5\)-cycle, each of the two remaining vertices \(r,s\) is adjacent to exactly one endpoint of this path and to neither its middle vertex. Hence again there exist \(x\in\{r,s\}\) and
indices \(i<j\) such that
$$
xv_i\notin E(W_6)
\qquad\text{and}\qquad
xv_j\in E(W_6).
$$
As before, \(xu\in E(W_6)\). Applying
Corollary~\ref{cor:corollary_betweeness} to \((u,v_i,v_j)\) and \(x\), we get
$$
p(u)<p(x)<p(v_j)\leq p(v_3),
$$
again contradicting the  fact that neither \(r\) nor \(s\) lies in the interval
\((p(u),p(v_3))\).

In both cases we obtain a contradiction. Therefore \(W_6\notin\cat\).
\end{proof}

\begin{theorem}\label{thm:co_sunlet3_not_cat}
The graph \( \overline{\emph{sunlet}_3} \) depicted in Figure~\ref{fig:graphs-section4}(d) is not in \(\cat\).
\end{theorem}

\begin{proof}
Let \(G=\overline{\emph{sunlet}_3}\). We denote its vertices by
$\{v_1,v_2,v_3,w_1,w_2,w_3\}$, where \(v_1,v_2,v_3\) induce a clique, \(w_1,w_2,w_3\) induce an independent set, and $v_iw_j\in E(G)$ if and only if  $i\neq j.$ Thus \(v_iw_i\notin E(G)\) for each \(i\).

Suppose, on the contrary that \(G\) admits an \([a,b]\)-difference labelling \(p\). Since \(G\) is twin-free, Observation~\ref{obs:same_label_twins} implies that \(p\) is injective. Relabelling if necessary, assume
$$
p(v_1)<p(v_2)<p(v_3).
$$

Now \(w_2\) is adjacent to \(v_1\) and \(v_3\), but not to \(v_2\). By Corollary~\ref{cor:corollary_betweeness}, we get $p(v_1)<p(w_2)<p(v_3).$
Hence either \(p(w_2)<p(v_2)\) or \(p(v_2)<p(w_2)\).

Suppose first that $p(v_1)<p(w_2)<p(v_2).$ The vertex \(w_3\) is adjacent to \(v_1\) and \(v_2\), but not to \(w_2\). Applying Corollary~\ref{cor:corollary_betweeness} to the triple \((v_1,w_2,v_2)\) and the vertex \(w_3\), we obtain $p(v_1)<p(w_3)<p(v_2).$ But \(v_3\) is adjacent to \(v_1\) and \(v_2\), and not adjacent to \(w_3\). Applying the same lemma again to \((v_1,w_3,v_2)\) and \(v_3\), we get $p(v_1)<p(v_3)<p(v_2),$ contradicting $p(v_2)<p(v_3)$.

The case $p(v_2)<p(w_2)<p(v_3)$  is symmetric. Indeed, using \(w_1\) in place of \(w_3\), the lemma first gives $p(v_2)<p(w_1)<p(v_3)$, and then, using \(v_1\), gives $p(v_2)<p(v_1)<p(v_3)$,
contradicting \(p(v_1)<p(v_2)\).

Therefore no such difference labelling exists, and hence \(G\notin\cat\).
\end{proof}

The result in Theorem~\ref{theo:w6} and Theorem~\ref{thm:co_sunlet3_not_cat} together with Theorem~\ref{theo:five_vertices}  show that the graphs $W_6$ and  $\overline{\mathrm{sunlet}_3}$ are graphs not in \cat~that  have the smallest number of vertices.

Finally in Figure~\ref{fig:summary}  we depict the schematic representation of the results presented in this section. 

\begin{figure}[h]
    \centering
\includegraphics[scale=0.6]{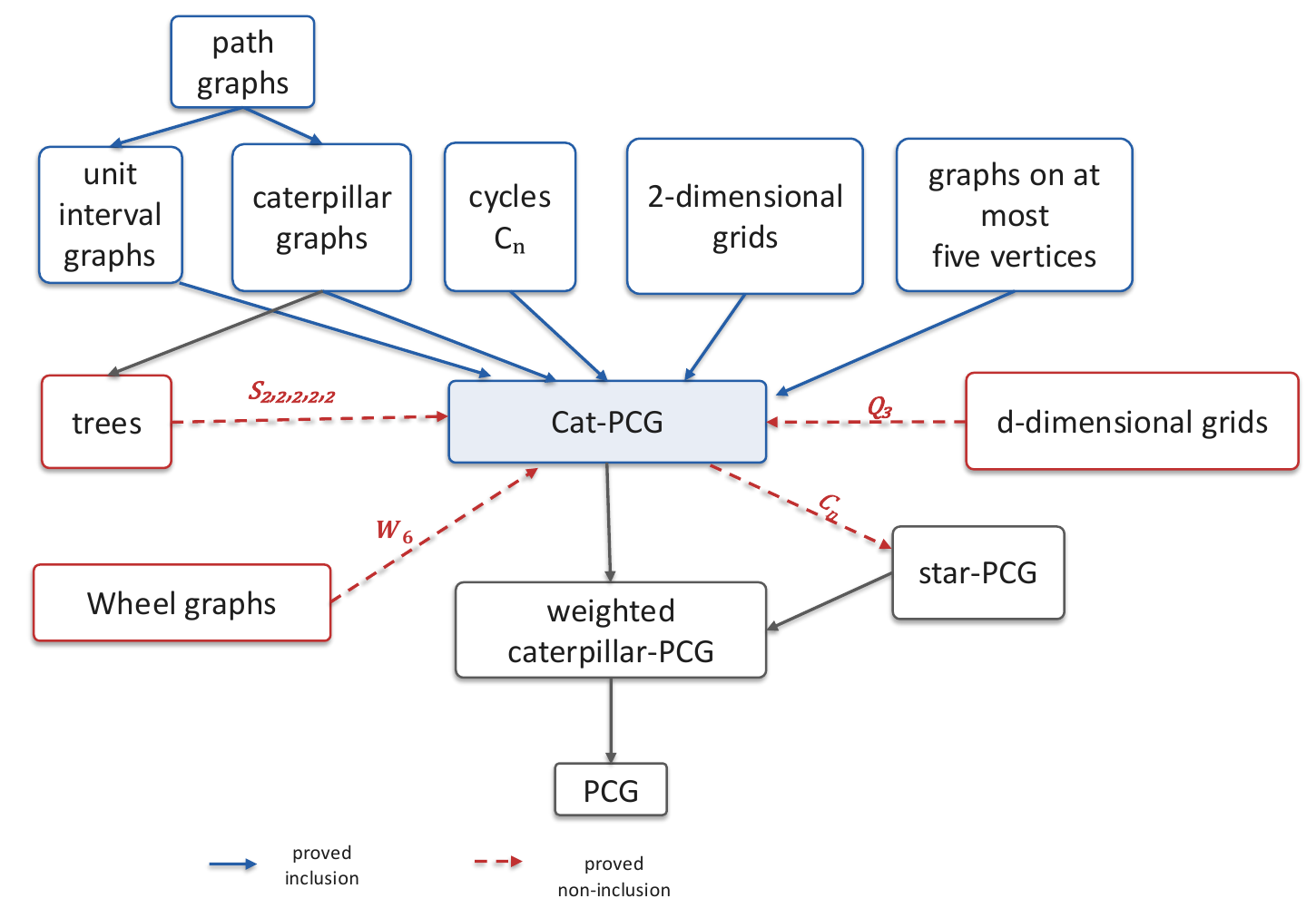}    
\caption{Schematic overview of the relationships between
\cat~and related graph classes.}
    \label{fig:summary}
\end{figure}

\section{Restrictions on Cat-PCG representations}\label{sec:restrictions} 

In this section we consider some particular subclasses of \cat s by restricting either the topology of the caterpillar or the interval used in the representation.

\subsection{Restricting the caterpillar: centipede witnesses}
A \emph{centipede} is a caterpillar in which every spine vertex has exactly one leaf.  We consider the \cat s for which the caterpillar is in fact a centipede.  Let \(P_n=v_1,v_2,\ldots,v_n\) be the path with $n$ vertices.  Given an integer \(r\geq 0\),  we denote by \(P_n^r\) the \(r\)-th power of \(P_n\), that is, the graph on vertex set
\([n]\) in which \(ij\) is an edge if and only if
\[
1\leq |i-j|\leq r.
\]
We also set \(P_n^r\) to be the empty graph when \(r\leq 0\).


The following characterization is essentially
\cite[Theorem~3.2]{calamoneri2014pairwise}; with our convention for $P_n^r$ when $r\leq0$, it can be written as follows.

\begin{theorem}[\cite{{calamoneri2014pairwise}}]\label{thm:centipede_characterization}
Let $G$ be a graph on $n$ vertices. Then $G$ has a centipede witness with interval
$[a,b]$ if and only if $G$ is isomorphic to $P_n^{\,b-2}-P_n^{\,a-3}$.
\end{theorem}

\begin{proof}
Suppose first that \(G\) has a centipede witness \(T\) with interval \([a,b]\). Let
\(s_1,\dots,s_n\) be the spine vertices in order, and let \(v_i\) be the unique leaf adjacent to
\(s_i\). Then, for \(i<j\), $d_T(v_i,v_j)=j-i+2$. Therefore $v_iv_j\in E(G)$
if and only if $j-i+2\in [a,b]$ or equivalently  $j-i\in [a-2,b-2]$.
Thus $G$ is isomorphic to $P_n^{\,b-2}-P_n^{\,a-3}$.

Conversely, suppose that $G$ is isomorphic to $P_n^{\,b-2}-P_n^{\,a-3}$. Build a centipede with spine
\(s_1,\dots,s_n\), and attach one leaf \(v_i\) to each \(s_i\). Then
$d_T(v_i,v_j)=|i-j|+2$,
so the interval \([a,b]\) realizes exactly the edges of \(G\). Hence \(G\) has a centipede witness.
\end{proof}

Notice that this class should not be confused with unit interval graphs. 
By Theorem~\ref{theo:unit_inerval}, every unit interval graph belongs to \(\cat\), and more precisely has a \([0,b]\)-difference labelling for some \(b\).  However, centipede witnesses with a general interval do not necessarily give unit interval graphs.  For example, a centipede on five leaves with interval \([4,5]\) gives the graph
with edges \(v_iv_j\) whenever \(|i-j|\in[2,3]\). This graph is the cycle
\(v_1-v_3-v_5-v_2-v_4-v_1\), that is, \(C_5\). Hence centipede witnesses can
produce graphs that are not unit interval graphs.

\subsection{Restricting the interval: the case $[a,a]$}
Here we characterize the class of \catp{a}{a} for all integer $a\geq 2$.

\begin{theorem}\label{theo:char_cat_2_2}
A graph \(G\) belongs to \(\catp{2}{2}\) if and only if \(G\) is a disjoint union of cliques.
\end{theorem}

\begin{proof}
By Theorem~\ref{thm:char}, \(G\in\catp{2}{2}\) if and only if \(G\) admits a
\([0,0]\)-difference labelling \(p\). In such a labelling,
$uv\in E(G)$ if and only if  $|p(u)-p(v)|=0$
or equivalently $p(u)=p(v).$ Thus vertices with the same label form a clique, and vertices with different
labels are non-adjacent. Therefore \(G\) is a disjoint union of cliques.
Conversely, if \(G\) is a disjoint union of cliques, assign the same label to all
vertices in each clique and assign distinct labels to different cliques. Then two
vertices have label difference \(0\) exactly when they belong to the same clique.
Hence this is a \([0,0]\)-difference labelling, and so \(G\in\catp{2}{2}\).
\end{proof}

Before moving to the more general case we need the following definition. A \emph{false-twin blow-up} of a graph \(H\) is obtained by replacing each vertex \(x\in V(H)\) by a nonempty independent set \(X_x\), and replacing each edge \(xy\in E(H)\) by all edges between \(X_x\) and \(X_y\).

\begin{theorem}\label{thm:exact_interval}
A graph \(G\) belongs to \(\catp{a}{a}\), for \(a\geq 3\), if and only if
\(G\) is a false-twin blow-up of a disjoint union of paths.
\end{theorem}

\begin{proof}
By Theorem~\ref{thm:char}, \(G\in\catp{a}{a}\) if and only if \(G\) admits a \([d,d]\)-difference labelling \(p\), where \(d=a-2>0\).

Suppose first that \(G\) has a \([d,d]\)-difference labelling \(p\). If two vertices
have the same label, then they are non-adjacent, since \(d>0\), and they have the
same neighbourhood. Hence vertices with the same label form independent sets of
false twins.

Contract each such independent set to one vertex. The resulting graph has an
injective \([d,d]\)-difference labelling. In such a labelling, a vertex with label
\(t\) can be adjacent only to vertices with labels \(t-d\) and \(t+d\). Therefore
every vertex has degree at most two. Moreover, no cycle can occur, since along
any path, the labels must move by \(+d\) or \(-d\), and an injective labelling cannot
return to the same label. Hence the contracted graph is a disjoint union of
paths. Thus \(G\) is obtained from a disjoint union of paths by replacing vertices
by independent sets of false twins.

Conversely, suppose that \(G\) is obtained from a disjoint union of paths by
replacing vertices by independent sets of false twins. Label each original path
with consecutive multiples of \(d\), leaving a gap of size larger than \(d\)
between different path components. Then adjacent vertices on a path have label
difference exactly \(d\), while vertices from different components have label
difference different from \(d\). Finally, give all vertices in the independent
set replacing a path vertex the same label. The resulting labelling is a
\([d,d]\)-difference labelling of \(G\). Hence \(G\in\catp{a}{a}\).
\end{proof}

We note that the characterization of Theorem~\ref{thm:exact_interval} yields a linear-time algorithm to decide whether a graph $G$ is in \catp{a}{a} for some $a \geq 3$.  Indeed, computing the partition of $V(G)$ into twin classes can be done in linear time using the \emph{partition refinement} approach~\cite{HABIB201041}.  In the same amount of time, it is easy to contract each false-twin class into a single vertex and verify that this results in a disjoint union of paths.

\subsection{Universal difference labellings}

We now consider graphs that admit a difference labelling for every positive
integer interval. If single point intervals (i.e. intervals of the type $[a,a]$) are allowed, the answer is completely
forced by the interval \([1,1]\), as shown by the next theorem.

\begin{definition}
We say that a graph \(G\) is \emph{universally difference-labellable} if, for
every pair of integers \(a,b\) with \(0<a\leq b\), \(G\) admits an
\([a,b]\)-difference labelling.
\end{definition}

\begin{theorem}\label{thm:universal_difference}
A graph \(G\) is universally difference-labellable if and only if \(G\) is a false-twin blow-up of a disjoint union of paths.
\end{theorem}

\begin{proof}
Suppose first that \(G\) is universally difference-labellable. Then, in
particular, \(G\) admits a \([1,1]\)-difference labelling. By
Theorem~\ref{thm:exact_interval}, this means that \(G\) is a false-twin blow-up
of a disjoint union of paths. Conversely, suppose that \(G\) is a false-twin blow-up of a disjoint union of paths. Again by Theorem~\ref{thm:exact_interval}, \(G\) admits a
\([1,1]\)-difference labelling \(p\). Let \(a,b\) be integers with
\(0<a\leq b\). Define $q(v)=b\,p(v)$ for every \(v\in V(G)\).

If \(uv\in E(G)\), then \(|p(u)-p(v)|=1\), and hence $|q(u)-q(v)|=b\in[a,b].$
If \(uv\notin E(G)\), then \(|p(u)-p(v)|\neq 1\). Since \(p\) is integer-valued,
we have either \(|p(u)-p(v)|=0\) or \(|p(u)-p(v)|\geq 2\). Therefore either
$|q(u)-q(v)|=0<a,$
or
$|q(u)-q(v)|\geq 2b>b.$
Thus \(q\) is an \([a,b]\)-difference labelling of \(G\). Since \(a,b\) were
arbitrary, \(G\) is universally difference-labellable.
\end{proof}

If instead we exclude single point intervals we have the following definition.

\begin{definition}
We say that a graph \(G\) is \emph{strictly universally difference-labellable}
if, for every pair of integers \(a,b\) with \(0<a<b\), \(G\) admits an
\([a,b]\)-difference labelling.
\end{definition}

\begin{observation}\label{obs:triangle_strict_universal}
No graph containing a triangle is strictly universally difference-labellable.
\end{observation}

\begin{proof}
Let \(G\) contain a triangle with vertices \(x,y,z\), and suppose that \(G\)
admits an \([a,b]\)-difference labelling \(p\). Assume, without loss of
generality, that \(p(x)<p(y)<p(z).\)
Since \(xy,yz\in E(G)\), we have
\(p(y)-p(x)\geq a\) and  \(p(z)-p(y)\geq a.\)
Therefore $p(z)-p(x)\geq 2a.$ But \(xz\in E(G)\), so we must also have
$p(z)-p(x)\leq b.$ Hence \(b\geq 2a\). Taking, for example, the interval \([2,3]\), this is impossible. Thus no graph containing a triangle can be strictly universally difference-labellable.
\end{proof}

\section{Recognition of \catp{a}{b}}\label{sec:recognition}
The question of whether Cat-PCG graphs can be recognized in polynomial
time appears to be difficult and is left open. In this section, we
consider \catp{a}{b} for fixed integers $0\leq a\leq b$. The main idea is to reconstruct the witness caterpillar locally by
encoding the positions at which its leaves are attached to the spine.
Thus, rather than working directly with trees, we work with sequences.

Since the distance between two distinct leaves is at least $2$, we set $A=a-2$  and $B=b-2$. Notice that we can assume that $b\geq a\geq 2$ as the distance between any two leaves in the caterpillar is at least 2.
By Theorem~\ref{thm:char}, the recognition problem can then be equivalently formulated in terms of $[A,B]$-difference labellings. Throughout this section, $(a,b)$ denote the parameters of the
caterpillar representation, while $(A,B)$ denote the corresponding
parameters of the sequence representation.

We show that one can decide whether a graph belongs to
\catp{a}{b} in time $O((n+1)^{4(b-2)})$ for $b\geq 3$. We now introduce the required definitions.


Let $S=(s_1,\dots,s_k)$ be a sequence. Its length is denoted by \(|S|=k\).
We write \(S[i]\) for its \(i\)-th element and define
$ set(S)=\{S[i]: i\in[k]\}.$ 
Slightly abusing notation, we may sometimes write $v \in S$ instead of $v \in set(S)$.  If an element $u$ occurs exactly once in $S$, we denote by
$pos_S(u)$ the unique index $i$ such that $S[i]=u$.


For a graph $G$, a \emph{partial \(G\)-sequence} is a sequence
$ S=(s_1,\dots,s_k)$ such that
$s_i \in V(G)\cup\{\bot\}$ for all $i\in[k],$
where $\bot\notin V(G)$ denotes an empty position. We say that $S$ is \emph{well-formed} if every vertex of \(G\) appears at most once in $S$. For a well-formed sequence \(S\) and vertices \(u,v\in set(S)\cap V(G)\), we define
$d_S(u,v)=|pos_S(u)-pos_S(v)|$.


For two partial \(G\)-sequences
$ S=(s_1,\dots,s_k), T=(t_1,\dots,t_\ell),$ 
their concatenation is defined as
$S\circ T = (s_1,\dots,s_k,t_1,\dots,t_\ell).$


Given a graph \(G\) and integers $0\leq A\leq B$, a partial $G$-sequence $S$ is \emph{valid with respect to \((G,A,B)\)} if the following hold:

\begin{enumerate}

\item\label{valid:well}
$S$ is well-formed;

\item\label{valid:in}
for all distinct \(u,v\in set(S)\cap V(G)\),
$ uv\in E(G) \iff
d_S(u,v)\in[A,B];
$ 

\item\label{valid:far}
for every \(u\in set(S)\cap V(G)\) such that $B < pos_S(u) \le |S|-B$,
all neighbors of \(u\)  appear in $S$, i.e.
$N_G(u)\subseteq set(S).$

\end{enumerate}


Condition~\ref{valid:in} enforces that the sequence respects the distance constraints of a Cat-PCG representation.
Condition~\ref{valid:far} is a locality condition: vertices sufficiently far from both ends of the sequence have already “seen” all their neighbors inside the sequence.


We use the following proposition.

\begin{proposition}
    Let $G$ be a graph without false twins if $a > 2$, or without true twins if $a \leq 2$. Then $G$ is in \catp{a}{b}~if and only if there exists a partial $G$-sequence $S$ with $V(G) \subseteq set(S)$ that is valid with respect to $(G, A, B)$.  
\end{proposition}

\begin{proof}
Suppose first that $G\in\catp{a}{b}$. By Theorem~\ref{thm:char}, $G$ admits an $[A,B]$-difference labelling $p$. By Observation~\ref{obs:same_label_twins}, and by the assumptions on $G$, $p$ is injective.

Let $r=\max_{v\in V(G)} p(v)$, and define $S=(s_1,\ldots,s_r)$ by placing each vertex $v$ in position $p(v)$ and $\bot$ in every unused position. Since $p$ is injective, $S$
is well-formed and contains all vertices of $G$.
Define a partial $G$-sequence $S$ by placing each vertex $v$ in
position $p(v)$ and $\bot$ in every unused position. Since $p$ is
injective, $S$ is well-formed and contains all vertices of $G$.
Moreover, $d_S(u,v)=|p(u)-p(v)|$, so Condition~\ref{valid:in} follows from the definition of an $[A,B]$-difference labelling. Condition~\ref{valid:far} holds immediately since $V(G)\subseteq set(S)$. Hence $S$ is valid.

Conversely, let $S$ be a valid partial $G$-sequence containing all
vertices of $G$, and define $p(v)=pos_S(v)$. By Condition~\ref{valid:in}, $p$ is an $[A,B]$-difference labelling of $G$. Hence, by Theorem~\ref{thm:char},
$G\in\catp{A+2}{B+2}=\catp{a}{b}$. 
\end{proof}



\begin{observation}\label{obs:subpcat}
Let $S$ be a valid partial \(G\)-sequence with respect to $(G,A,B)$.
Then every contiguous subsequence of $S$ is also valid.
\end{observation}

\begin{proof}
Well-formedness and Condition~\ref{valid:in} are preserved immediately.
We verify Condition~\ref{valid:far}. Let \(S'\) be a contiguous subsequence of \(S\), and let $ u\in set(S')\cap V(G)$ 
satisfy $ B < pos_{S'}(u) \le |S'|-B.$  Then $u$ is also at distance at least $B$ from both ends of $S$, so Condition~\ref{valid:far} for $S$ implies $N_G(u) \subseteq set(S)$.  Moreover, by Condition~\ref{valid:in} every neighbor $v$ of $u$ satisfies $d_S(u,v)\leq B$. Since $u$ is at distance at least $B$ from both ends of $S'$, every such neighbor also belongs to $S'$. Hence $N_G(u)\subseteq set(S')$, and Condition~\ref{valid:far} holds.
\end{proof}


Notice that if $b=2$, then $\catp{a}{2}=\catp{2}{2}$,
which is characterized by Theorem~\ref{theo:char_cat_2_2}. Hence, in the remainder of the algorithmic construction, we assume $b\geq3$, and therefore
$B\geq1$.

\begin{lemma}[Locality of neighbors]\label{lem:locality}
Let \(S\) be a valid partial \(G\)-sequence with respect to \((G,A,B)\).
Let $uv \in E(G)$ and suppose that $u, v \in S$. Then
$|pos_S(u)-pos_S(v)|\le B.$

Moreover, suppose that \(S\) is partitioned into consecutive blocks
$S=S_1\circ S_2\circ\cdots \circ S_l$
where $|S_i| = B$ for each $i \in [l-1]$ and $|S_l| \leq B$.
If $u \in S_i$, then 
\[
v \in S_{i-1} \circ S_i \circ S_{i+1}
\]
where here we consider that $S_0$ and $S_{l+1}$ are empty sequences. 
\end{lemma}

\begin{proof}
Since \(uv\in E(G)\) and $u,v \in S$, Condition~\ref{valid:in} gives
$ d_S(u,v)\in[A,B],$
hence $d_S(u,v)\le B$, proving the first claim.

For the block statement, assume $v\in S_j$ with $j\le i-2$.
Then between $u$ and $v$ there is at least one full block of size $B$, hence $|pos_S(u)-pos_S(v)| \ge B+1,$  contradicting the previous bound. The case $j\ge i+2$ is symmetric. Thus $v$ must lie in $S_{i-1}\circ S_i\circ S_{i+1}$.
\end{proof}

We now show that deciding whether a graph $G$ is in \catp{a}{b}~reduces to the problem of finding a path in a directed graph, in which such paths represent the concatenation of valid partial $G$-sequences.

Let $\mathcal S$ be the set of all partial $G$-sequences of length $B$ that are valid with respect to $(G, A, B)$.    
We denote by $\tnil$ the partial $G$-sequence of length $B$ with $set(\tnil) = \{\bot\}$. Note that $\tnil \in \mathcal S$. 

We define $H := H(G)$ as the directed graph with vertex set
\begin{align*}
    V(H) &= \{ (S_1, S_2, S_3) : S_1, S_2, S_3 \in \mathcal S \mbox{ and $S_1 \circ S_2 \circ S_3$ is valid w.r.t. $(G, A, B)$} \}.
\end{align*}
Then, for each $(S_1, S_2, S_3) \in V(H)$, add a directed edge from $(S_1, S_2, S_3)$ to $(S_2, S_3, S_4)$, for each $S_4 \in \mathcal S$ such that $(S_2, S_3, S_4) \in V(H)$. 

The intuition is that paths in $H$ correspond to partial $G$-sequences,
and the goal is to find such a path that corresponds to a valid sequence
with all of $V(G)$ that starts and ends with $\tnil$. One issue is
that the vertices represent local sequences $S_1\circ S_2\circ S_3$,
whereas the desired sequence needs to include all vertices. However,
if $G$ is connected, Condition~\ref{valid:far} ensures that if a vertex
of $G$ appears in the middle block, then all of its neighbors appear in
the corresponding three-block window, and in turn the neighbors of
neighbors, and so on, ensuring all vertices are present before we finish
with $\tnil$. Another issue is that sequences corresponding to paths
of $H$ may contain vertices multiple times, but repetitions can be
removed in a post-processing step.

\begin{theorem}\label{thm:abpath}
Let $G$ be a connected graph. There exists a valid partial $G$-sequence $S$
such that $V(G)\subseteq set(S)$ if and only if $H$ contains a directed path
from a vertex of the form $(\tnil,S_1,S_2)$ with $S_1\neq\tnil$, to a vertex of the form $ (S_{k-1},S_k,\tnil).$ 
\end{theorem}

\begin{proof}

\textbf{(Forward direction).}
Suppose that $G$ admits a partial $G$-sequence $S$ with $V(G) \subseteq set(S)$.   By deleting leading $\bot$ symbols and adding $\bot$ symbols at the end if necessary, we may assume that the first element of $S$ is not $\bot$, and that $|S|$ is a multiple of $B$ and at least $2B$. So denote $S=(s_1,\dots,s_{Bp})$, where $p$ is some integer. For each \(i\in[p]\), define the block $S_i=(s_{(i-1)B+1},\dots,s_{iB}).$ By Observation~\ref{obs:subpcat}, each \(S_i\) is valid, hence \(S_i\in \mathcal S\).

For every \(i\in\{2,\dots,p-1\}\), the concatenation $S_{i-1}\circ S_i\circ S_{i+1}$  is a contiguous subsequence of \(S\), and therefore valid. Hence each triple \((S_{i-1},S_i,S_{i+1})\) is a vertex of \(H\).

We now handle the boundary. Consider the sequence $\tnil\circ S_1\circ S_2$.
It is well-formed, since \(\tnil\) introduces no vertices of \(G\),
and distances among vertices in \(S_1\circ S_2\) are unchanged, so Condition~\ref{valid:in} holds. We verify Condition~\ref{valid:far}.
Let \(u\in S_1\) satisfy
$$
B<pos_{\tnil\circ S_1\circ S_2}(u)\le 2B.
$$
Then $u$ is at distance at least $B$ from both ends of the window
and therefore every neighbor of $u$ in $G$ lies within distance at most $B$
in the original sequence $S$.
Since no vertex of \(G\) can lie outside the first two blocks while being adjacent to such a vertex,
we conclude that
$ N_G(u)\subseteq set(S_1\circ S_2).$ 
Thus $(\tnil,S_1,S_2)\in V(H)$. A symmetric argument shows that \((S_{p-1},S_p,\tnil)\in V(H)\). Therefore $H$ contains the directed path
$$
(\tnil,S_1,S_2),
(S_1,S_2,S_3),
\dots,
(S_{p-1},S_p,\tnil).
$$
This is the desired path, since $S_1 \neq \tnil$ because it contains $s_1 \neq \bot$, in particular.

\medskip

\noindent
\textbf{(Backward direction).}
Assume that $H$ contains a directed path
$$
(\tnil,S_1,S_2),(S_1,S_2,S_3),\ldots,(S_{k-1},S_k,\tnil),
$$
with $S_1\neq\tnil$.
Set $S_0=S_{k+1}=\tnil$ and, for every $i\in[k]$, let $W_i=S_{i-1}\circ S_i\circ S_{i+1}$.
By the definition of $H$, each $W_i$ is valid with respect to $(G,A,B)$.

\medskip
\noindent
\textbf{Step 1: Coverage of vertices.}
Let $r$ be the first non-$\bot$ vertex appearing in $S_1$.
We prove by induction on $d_G(r,v)$ that every vertex $v\in V(G)$
appears in some block $S_i$.

The claim is immediate for $v=r$. Let $v\neq r$, and let $w$ be the
neighbor of $v$ preceding it on a shortest $r$--$v$ path. By induction,
$w\in S_i$ for some $i$. Since $S_i$ is the middle block of $W_i$, $B<pos_{W_i}(w)\leq 2B.$
As $W_i$ is valid, Condition~\ref{valid:far} gives $N_G(w)\subseteq set(W_i).$
Since $v\in N_G(w)$, the vertex $v$ appears in one of
$S_{i-1},S_i,S_{i+1}$. Hence every vertex of $G$ appears in at least
one block.

\medskip
\noindent
\textbf{Step 2: First occurrences.}
For each $v\in V(G)$, let $i(v)=\min\{i:v\in S_i\}.$
Starting from $S_1\circ\cdots\circ S_k$, keep only the first occurrence
of each vertex and replace every later occurrence by $\bot$.
Let $S^*$ be the resulting sequence. Then $S^*$ is well-formed and $V(G)\subseteq set(S^*).$

\medskip
\noindent
\textbf{Step 3: Validity of $S^*$.}
Condition~\ref{valid:well} holds by construction. Condition~\ref{valid:far}  also holds, since
$S^*$ contains every vertex of $G$. It remains to verify Condition~\ref{valid:in}. Let $u,v\in V(G)$ be distinct, and assume without loss of generality that $i(u)\leq i(v)$.

Suppose first that $uv\in E(G)$. The first occurrence of $u$ lies in
the middle block $S_{i(u)}$ of the valid sequence $W_{i(u)}$. Hence,
by Condition~\ref{valid:far}, $v$ appears in $W_{i(u)}$. Since $i(v)\geq i(u)$, its
first occurrence lies in either $S_{i(u)}$ or $S_{i(u)+1}$. Therefore
the first occurrences of $u$ and $v$ both belong to $W_{i(u)}$.
By Condition~\ref{valid:in} for $W_{i(u)}$, $d_{S^*}(u,v)\in[A,B].$

Conversely, suppose that $d_{S^*}(u,v)\in[A,B].$
Then $d_{S^*}(u,v)\leq B$. Since every block has length $B$, the first
occurrences of $u$ and $v$ lie in the same block or in two consecutive
blocks. Hence they both occur in some valid window $W_i$. By
Condition~\ref{valid:in} for $W_i$, $uv\in E(G).$ Thus
$$uv\in E(G)
\quad\Longleftrightarrow\quad
d_{S^*}(u,v)\in[A,B],
$$
and therefore $S^*$ is a valid partial $G$-sequence containing all
vertices of $G$.
\end{proof}

\begin{lemma}\label{lem:twins}
Let \(G\) be a graph and let \(a,b\) be integers. Suppose that there exist distinct vertices \(u,v\in V(G)\) such that either:
\begin{itemize}
\item \(a>2\) and \(u,v\) are false twins, or
\item \(a\le 2\le b\) and \(u,v\) are true twins.
\end{itemize}

Then
$$
G\in\mathrm{CatPCG}(a,b)
\iff
G-v\in\mathrm{CatPCG}(a,b).
$$
\end{lemma}

\begin{proof}
The forward implication is immediate.
Conversely, suppose that \(G-v\) admits a witness caterpillar \(T\).
Let \(w\) be the neighbor of \(u\) in \(T\), and attach \(v\) as a new leaf adjacent to \(w\).
If \(a>2\), then \(d(u,v)=2\notin[a,b]\), so \(u\) and \(v\) are nonadjacent.
If \(a\le 2\le b\), then \(d(u,v)=2\in[a,b]\), so \(u\) and \(v\) are adjacent.
Moreover, for every other leaf $z$, $d(v,z)=d(u,z)$, so $u$ and $v$
have the same neighbors outside the pair $\{u,v\}$. Thus the resulting
caterpillar is a witness for $G$.
\end{proof}


\begin{theorem}\label{thm:recognition}
Let $G$ be a graph with $n$ vertices, and let $0\leq a\leq b$ be fixed
integers. Then one can decide in polynomial time whether
$G\in\catp{a}{b}$. In particular, for $b\geq 3$, the problem can be
solved in time $O((n+1)^{4(b-2)}).$
\end{theorem}

\begin{proof}
If $b<2$, then $G\in\catp{a}{b}$ if and only if $G$ is edgeless.
If $b=2$, then the result follows from Theorem~\ref{theo:char_cat_2_2}.
Thus, assume $b\geq3$. We first repeatedly apply Lemma~\ref{lem:twins}:
\begin{itemize}
    \item remove false twins when $a>2$;
    \item remove true twins when $a=2$.
\end{itemize}
Each removal strictly decreases the number of vertices, and the resulting
instance is equivalent by Lemma~\ref{lem:twins}.

After the twin reduction, we apply Theorem~\ref{thm:abpath} by constructing
$H:=H(G)$. Recall that $\mathcal S$ is the set of valid partial
$G$-sequences of length $B$. Since each position contains either a vertex
or $\bot$, $|\mathcal S|\leq(n+1)^B.$ Therefore
$|V(H)|\leq(n+1)^{3B}.$
Each vertex $(S_1,S_2,S_3)$ has at most $(n+1)^B$ outgoing transitions
to a vertex $(S_2,S_3,S_4)$, and hence $|E(H)|\leq(n+1)^{4B}.$

By Theorem~\ref{thm:abpath}, it suffices to determine whether there is a
path from some vertex of the form $(\tnil,S_1,S_2)$ to some vertex of the
form $(S_{k-1},S_k,\tnil)$. We add a new source vertex adjacent to all
possible starting vertices and run a single BFS from this source. This
takes time linear in the size of $H$. Thus the recognition problem can be solved in time $ O((n+1)^{4B})=O((n+1)^{4(b-2)}).$
\end{proof}


\section{Graph operations}\label{sec:graph_operations}
In this section, we study the behaviour of
\(\mathrm{Cat\text{-}PCG}\) under some standard graph operations. Recall that by Theorem~\ref{thm:char}, it is enough to study these
operations in terms of difference labellings.

The full class \(\mathrm{Cat\text{-}PCG}\) is also hereditary and
closed under adding isolated vertices. However, allowing different
intervals for different components does not preserve closure under
disjoint union.

\begin{theorem}
The class \(\mathrm{Cat\text{-}PCG}\) is not closed under disjoint
union.
\end{theorem}

\begin{proof}
Both \(K_4\) and \(C_6\) belong to
\(\mathrm{Cat\text{-}PCG}\). We show that $K_4\mathbin{\dot\cup}C_6
\notin \mathrm{Cat\text{-}PCG}$. Suppose, on the contrary, that \(K_4\mathbin{\dot\cup}C_6\) admits an \([a,b]\)-difference labelling \(p\). Notice first that \(a>0\). Indeed, graphs admitting a \([0,b]\)-difference labelling are unit interval graphs, whereas \(C_6\) is not a unit interval graph. Let $x_1,x_2,x_3,x_4$ be the vertices of the $K_4$ component,
ordered so that $p(x_1)<p(x_2)<p(x_3)<p(x_4)$ (note that for distinct $i, j \in [4]$, $p(x_i) \neq p(x_j)$ since $x_i, x_j$ are neighbors and $a > 0$). Since every pair is adjacent, each consecutive difference is at least \(a\), while $p(x_4)-p(x_1)\leq b$.  Therefore $3a\leq p(x_4)-p(x_1) \leq b$,
and hence $b\geq 3a$.

Since $C_6$ is twin-free by Observation~\ref{obs:same_label_twins} $p$ is injective on the $C_6$ component. Let $c_0$ be the vertex of minimum label and, up to symmetry, assume $p(c_1)<p(c_5)$. By Lemma~\ref{lem:betweenness_obstruction} $c_1$ and $c_5$ must be consecutive in the label order. A direct check of the possible positions of $c_2,c_3,c_4$, again using Lemma~\ref{lem:betweenness_obstruction}, leaves only the two orders
$p(c_0)<p(c_1)<p(c_5)<p(c_2)<p(c_4)<p(c_3)$ and
$p(c_0)<p(c_1)<p(c_5)<p(c_4)<p(c_2)<p(c_3).$.  The second order is impossible. Since $c_4c_5\in E(C_6)$ and $p(c_1)<p(c_5)$, we have
$p(c_4)-p(c_1)>a.$ Moreover, since $p(c_4)<p(c_2)$ and $c_1c_2\in E(C_6)$,
$p(c_4)-p(c_1)<p(c_2)-p(c_1)\leq b.$ Thus $c_1c_4\in E(C_6)$, a contradiction.
Therefore $p(c_0)<p(c_1)<p(c_5)<p(c_2)<p(c_4)<p(c_3).$


Set
\[
r=p(c_5)-p(c_1),\qquad
s=p(c_2)-p(c_5),\qquad
t=p(c_4)-p(c_2).
\]
The pairs \(c_1c_5\), \(c_5c_2\), and \(c_2c_4\) are nonedges.
Using the label order above and the edges \(c_0c_5\), \(c_1c_2\),
and \(c_2c_3\), respectively, we obtain
\[
\begin{aligned}
r&=p(c_5)-p(c_1)<p(c_5)-p(c_0)\leq b,\\
s&=p(c_2)-p(c_5)<p(c_2)-p(c_1)\leq b,\\
t&=p(c_4)-p(c_2)<p(c_3)-p(c_2)\leq b.
\end{aligned}
\]
Since \(r,s,t\) correspond to nonedges and are smaller than \(b\),
they must satisfy \(r<a, s<a, t<a\). Since \(c_1c_2\) is an edge, \(r+s=p(c_2)-p(c_1)\geq a\).
Furthermore, \(c_1c_4\) is a nonedge and \(p(c_4)-p(c_1)=r+s+t>a\).
Therefore this difference must be larger than \(b\), and hence
\(r+s+t>b.\)
Consequently, \(b<r+s+t<3a\),
so \(b<3a\) contradicting \(b\geq3a\).

Therefore \( K_4\mathbin{\dot\cup}C_6\notin\mathrm{Cat\text{-}PCG},
\) and the class is not closed under disjoint union.
\end{proof}
The class $\mathrm{Cat\text{-}PCG}$ is not closed under some other natural graph operations either.

\begin{proposition}
The class \(\mathrm{Cat\text{-}PCG}\) is not closed under:
\begin{enumerate}
    \item adding a universal vertex;
    \item taking joins;
    \item taking complements.
\end{enumerate}
\end{proposition}

\begin{proof}
The cycle $C_5$ belongs to $\mathrm{Cat\text{-}PCG}$, whereas
adding a universal vertex produces $W_6$, which does not belong to
the class. Since $W_6=K_1\vee C_5$, the class is not closed under joins.

Finally, by Theorem~\ref{theo:cycles}, $C_5$ admits a $[3,5]$-difference
labelling, and $K_1$ trivially admits the same interval. Hence, by
Proposition~\ref{prop:close},
$K_1\mathbin{\dot\cup}C_5\in\mathrm{Cat\text{-}PCG}$.
Since $C_5$ is self-complementary,
$$
\overline{K_1\mathbin{\dot\cup}C_5}
\cong K_1\vee C_5
\cong W_6.
$$
Therefore, the class is not closed under complements.
\end{proof}

\section{Open problems}\label{sec:concl}

We conclude with some natural questions that arise from our results.
To the best of our knowledge, all the problems below are still open.

\begin{problem}[Recognition]
Can \cat be recognized in polynomial time?
\end{problem}

\begin{problem}[Disjoint unions]
For which pairs of graphs \(G\) and \(H\) is
\(G\mathbin{\dot\cup}H\)  in \cat?
\end{problem}

\begin{problem}[Grids with diagonals]
For which values of \(m\) and \(n\) is the strong product
\(P_m\boxtimes P_n\) in \cat ?
\end{problem}

The graph \(P_m\boxtimes P_n\) can be viewed as a rectangular grid in
which both diagonals are added to every square. Since ordinary
two-dimensional grids belong to the class, it is natural to ask whether
the same construction can be adapted when diagonal edges are present.

\begin{problem}[Toroidal grids]
For which values of \(m\) and \(n\) is the toroidal grid
\(C_m\square C_n\) in \cat?
\end{problem}

The construction used for ordinary grids relies on the boundary of the grid and on a suitable ordering of its diagonals. A toroidal grid has no boundary, so the same construction cannot be applied directly. It would already be interesting to study the families
\(C_3\square C_n\) and \(C_4\square C_n.\)

\begin{problem}[Star-PCGs]
Is every (weighted) star-PCG in \cat?
\end{problem}

Notice that the reverse inclusion does not hold. Indeed, every cycle
belongs to $\mathrm{Cat\text{-}PCG}$, whereas $C_n$ is not a
star-PCG for every $n\geq 5$. The remaining question is therefore whether
\[
\mathrm{star\text{-}PCG}
\subseteq
\mathrm{Cat\text{-}PCG}.
\]
If the inclusion is false, it would be interesting to determine the
smallest star-PCG that is not an unweighted caterpillar PCG.

\begin{problem}[Twin-width]
What is the smallest constant \(c\) such that every unweighted
caterpillar PCG has twin-width at most \(c\)?
\end{problem}

We proved that every graph in the class has twin-width at most \(10\).
It is not clear whether this bound is tight. A first step would be to
find examples with large twin-width or to improve the contraction
sequence used in the proof.

\begin{problem}[Forbidden induced subgraphs]
Determine the minimal forbidden induced subgraphs for \cat.
\end{problem}

Since the class is hereditary, it is completely determined by a family (potentially infinite) of  minimal forbidden induced subgraphs. Several concrete questions arise:
\begin{itemize}
    \item What are all minimal obstructions on $k$ vertices, with $k\geq6$?
    \item Are there infinitely many minimal obstructions?
    \item Does the class admit a finite forbidden induced-subgraph
characterization?
\end{itemize}

The nonmembers identified in this paper provide a first collection of
possible obstructions, but a systematic description is still missing.

\section{AI Disclosure}
We used AI to enhance readability and for to help with the drawing of Figures 2 and 3. No AI tools have contributed to the mathematical content of the manuscript.

\bibliographystyle{alpha}
\bibliography{references}
\end{document}